\documentclass[12pt,a4paper]{article}
\usepackage[top=2cm, bottom=2cm, left=2cm, right=2cm]{geometry}
\usepackage{mathptmx} 
\usepackage{microtype} 
\usepackage{adjustbox}
\usepackage{amssymb,amsthm,amsmath,latexsym}
\usepackage{authblk}
\usepackage[symbol]{footmisc}
\usepackage{graphicx}
\usepackage[table]{xcolor} 
\usepackage{colortbl} 
\usepackage{booktabs} 
\usepackage{epstopdf} 
\usepackage[T1]{fontenc}
\usepackage[utf8]{inputenc}
\usepackage{multirow}
\usepackage{pdfpages}
\usepackage{subcaption}
\usepackage{float}
\usepackage{orcidlink}
\usepackage{xurl}
\newtheorem {theorem}{Theorem}

\numberwithin{equation}{section}

\begin{document}
\title{ Unit Shiha Distribution and its Applications to Engineering and Medical Data}
\newcommand{\orcidauthorA}{\orcidlink{0000-0001-9539-9488}}
\author{
\begin{flushleft}
F. A. Shiha$^{}$\orcidauthorA{}\\
{\small
$^{}$Department of Mathematics, Faculty of Science, Mansoura University, 35516 Mansoura, Egypt\\

$^{}$Correspondence: Email: fshiha@mans.edu.eg
}
\end{flushleft}
}
\date{}
\maketitle

\begin{abstract}
There is a growing need for flexible statistical distributions that can accurately model data defined on the unit interval. This paper introduces a new unit distribution, termed the unit Shiha (USh) distribution, which is derived from the original Shiha (Sh) distribution through an inverse exponential transformation. The probability density function of the USh distribution is sufficiently flexible to model both left- and right-skewed data, while its hazard rate function is capable of capturing various failure-rate patterns, including increasing, bathtub-shaped, and J-shaped forms. Several statistical properties of the proposed distribution are investigated, including moments and related measures, the quantile function, entropy, and stress-strength reliability. Parameter estimation is carried out using the maximum likelihood method, and its performance is evaluated through a simulation study. The practical usefulness of the USh distribution is demonstrated using four real-life data sets, and its performance is compared with several well-known competing unit distributions. The comparative results indicate that the proposed model fits the data better than the competitive models applied in this study.

\end{abstract}

\textbf{keywords:} Unit Shiha distribution; Reliability; Maximum likelihood estimation; Data analysis.
\newline
\textbf{Mathematics Subject Classification:} 60E05, 62H12, 62N05, 62H10.

\maketitle

\section {Introduction}
In recent years, there has been increasing interest in developing new probability distributions for practical data analysis. Most existing distributions are defined on the interval $(0, \infty)$, which makes them unsuitable for modeling data that are bounded between 0 and 1. Such data commonly arise in practice, including standardized measurements, proportions, and percentages. Therefore, there is a clear need for simple and flexible probability distributions defined on this unit interval.
One common and effective approach for constructing unit distributions is to apply suitable transformations to positive-valued random variables. Transformations such as  $X = Y/(1 + Y)$ and $X= e^{-Y}$, where $Y$ defined on $(0,\infty)$, are widely used to map the support to the unit interval.
Based on this idea, numerous unit distributions have been proposed in the literature for modeling data on (0,1) across various application areas.
 Notable examples include the unit Xgamma distribution \cite{hashm}, log-XLindley distribution \cite{eliwa}, unit exponential distribution \cite{hasan}, unit Weibull distribution \cite{mazu}, unit exponentiated Weibull distribution \cite{ammar},  unit teissier distribution \cite{krish}, unit exponentiated half-logistic distribution \cite{hassan}, unit Haq distribution \cite{alz}, log-cosine-power unit distribution \cite{log}, unit-exponentiated Lomax distribution \cite{fayomi}, and unit modified Burr-III distribution \cite{ahsan}. A comprehensive review of continuous probability distributions on (0,1) can be found in \cite{rev}.

Shiha \cite{fshiha} introduced the Shiha distribution (ShD) as a two-parameter mixture of  $\text{Exp}(\omega)$, $\text{Exp}(2\omega)$, and $\text{Gamma}(2,2\omega)$ distributions. The ShD exhibits flexible hazard rate patterns, which makes it suitable for modeling a wide range of lifetime data on $ (0, \infty)$.
Motivated by these properties, the main objective of this study is to introduce a unit version of the ShD, referred to as the unit Shiha distribution (UShD), obtained through the transformation $X= e^{-Y}$, where $Y$ follows the Sh distribution.
 The proposed UShD extends the applicability of the original ShD to the unit interval while retaining its flexibility and analytical tractability. It is capable of modeling both left- and right-skewed data and capturing various failure-rate behaviors, including increasing, bathtub-shaped, and J-shaped patterns. These features make the UShD a useful and flexible model for analyzing bounded data arising in reliability, medical, and engineering applications.

 The remainder of this paper is organized as follows. Section 2 introduces the proposed distribution and its hazard rate and stress–strength reliability functions. Section 3 presents the main statistical properties, including moments and related measures, the quantile function, and entropy. Parameter estimation via maximum likelihood and a simulation study are discussed in Sections 4 and 5, respectively. Section 6 illustrates the performance of the proposed model using four real-life data sets and comparisons with several competing distributions. Finally, conclusions are presented in Section 7.

\section{The Unit Shiha Distribution }
The probability density function (PDF) of Shiha distribution \cite{fshiha} is
defined as
\begin{equation}
  f(y) = \frac{\omega}{\omega+3\,\eta}\left[\omega + (2\,\eta + 8\,\omega\,\eta y)e^{-\omega y}\right] e^{-\omega y}, \quad y\geq 0, \, \omega > 0, \, \eta > 0,
\end{equation}
the corresponding cumulative distribution function (CDF) is
\begin{equation}
  F(y) = 1 - \frac{1}{\omega+3\,\eta}\left[\omega + (3\,\eta + 4\,\omega \,\eta y)e^{-\omega y}\right] e^{-\omega y}, \quad y\geq 0, \, \omega > 0, \, \eta > 0.
\end{equation}
Let $X=e^{-Y}$, since $Y\geq 0$, we have $X\in (0,1]$ and $Y=-\ln X$, then the CDF of the transformed (unit) distribution is derived as
\begin{align}\label{tcdf}
  F(x) &= P(X\leq x)=P( e^{-Y} \leq x) \nonumber \\
   & =P(Y\geq -\ln x)=1-P(Y<-\ln x) \nonumber \\
   & =1-F_Y(-\ln x),
\end{align}
with the associated PDF
\begin{equation}\label{tpdf}
  f(x) =f_Y(-\ln x) \, \left|\frac{dy}{dx}\right| = \frac{1}{x} \, f_Y(- \ln x).
\end{equation}
Based on (\ref{tpdf}) and (\ref{tcdf}), the PDF and CDF of the unit Shiha distribution ($\text{USh}(\omega, \eta$)) are given by
\begin{equation}\label{pdf}
f(x;\omega, \eta)=\frac{\omega}{\omega+3\eta}\left[\omega x^{\omega-1}+(2\eta-8\omega\eta\ln x)x^{2\omega-1}\right],\qquad 0<x\le1, ,\ \omega>0,\ \eta>0.
\end{equation}
\begin{equation} \label{cdf}
F(x;\omega, \eta)=\frac{1}{\omega+3\eta}\left[\omega x^{\omega}+(3\eta-4\omega\eta\ln x)x^{2\omega}\right],\qquad 0< x < 1,
\end{equation}
with $F_X(x)=0$ for $x \le 0$ and $F_X(x)=1$ for $x \ge 1$. Note that the $\text{USh}(\omega,\eta)$ distribution reduces to the $\text{Beta}(\omega, 1)$ distribution when $\eta=0$.

The hazard rate function of UShD is
\begin{equation}
h(x) = \frac{\omega x^{\omega-1} \left[ \omega + 2\eta x^{\omega} (1 - 4\omega \ln x) \right]}
{\omega + 3\eta - x^{\omega} \left[ \omega + \eta x^{\omega} (3 - 4\omega \ln x) \right]},
\quad 0 < x < 1.
\end{equation}
Figure \ref{fig:pdfha} illustrates the PDF and the hazard rate function of the UShD for different values of the parameters $\omega$ and $\eta$. It is observed that the PDF of the USh distribution is flexible enough to model both left-skewed and right-skewed data. Moreover, the hazard rate function exhibits a variety of shapes, including bathtub-shaped, increasing, and J-shaped forms. This flexibility indicates that the proposed distribution is capable of modeling a wide range of lifetime data with different failure behaviors, making it suitable for applications in both engineering and medical fields.
\begin{figure}[htbp]
  \centering
 \includegraphics[width=14cm, height=5cm]{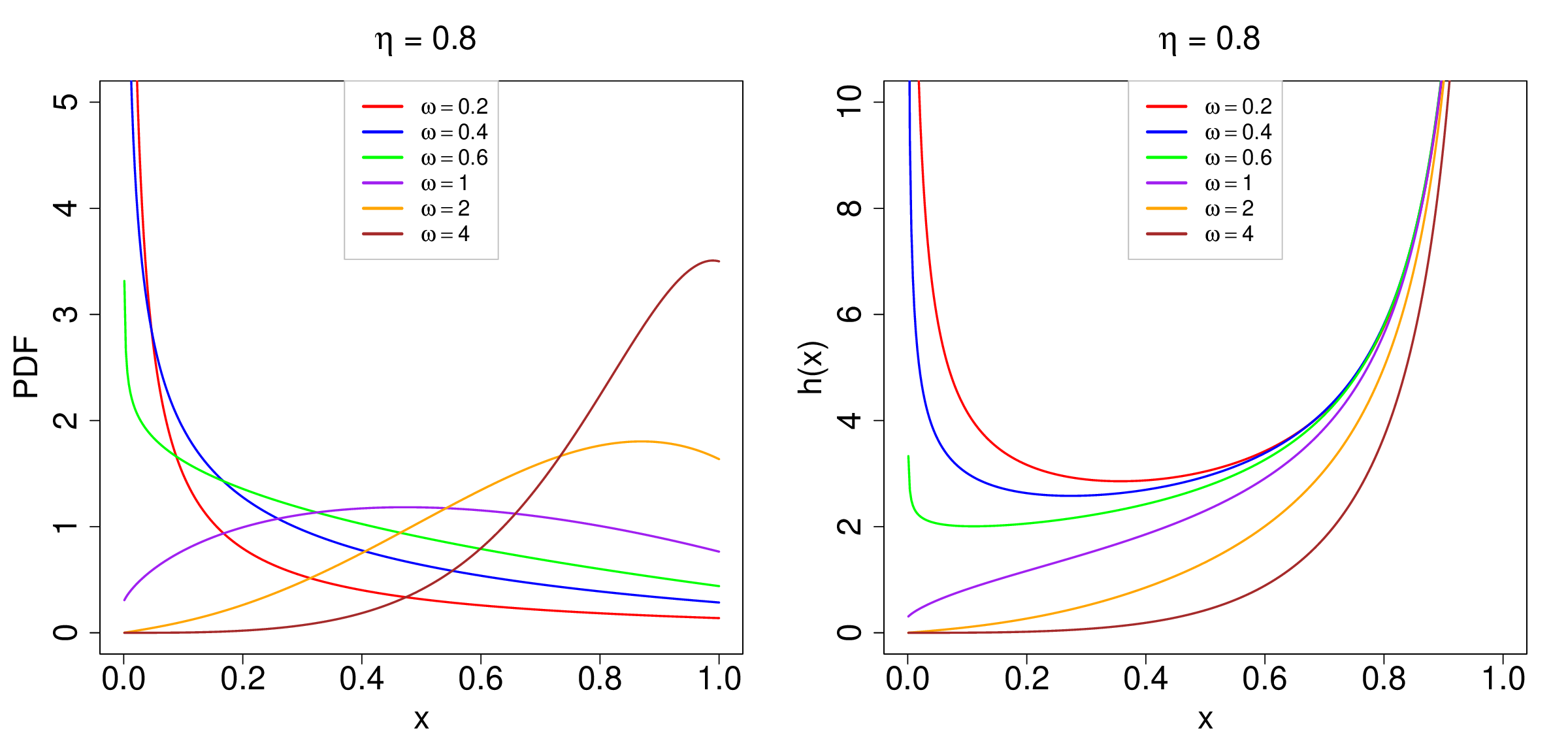}

\caption{The PDF and hazard functions of the UShD at different parameter values.}
 \label{fig:pdfha}
\end{figure}
\subsection{Stress-Strength Reliability}
The stress-strength reliability model is a fundamental concept in reliability engineering and statistical analysis that measures the probability that a system's strength exceeds the applied stress,  with failure occurring when the stress exceeds the strength. Formally, for two independent random variables $X_1$ and $X_2$, where $X_1$ represents the "strength" of a component or system and $X_2$ represents the "stress" imposed upon it, the stress-strength reliability is defined as $R = P(X_2< X_1)$.

Let $X_1$ and $X_2$ be independent USh random variables with parameters $(\omega_1,\eta_1)$ and $(\omega_2,\eta_2)$, respectively.
Then the stress–strength reliability parameter $R$ is given by
\[
R=\int_0^1 f(x; \omega_1,\eta_1)F(x; \omega_2,\eta_2)\,dx.
\]
Substituting from (\ref{pdf}) and (\ref{cdf}), we obtain
\[
\begin{aligned}
R
&=\frac{\omega_1}{(\omega_1+3\eta_1)(\omega_2+3\eta_2)}
\int_0^1
\Big[
\omega_1 x^{\omega_1-1}
+\left(2\eta_1-8\omega_1\eta_1\ln x\right)x^{2\omega_1-1}
\Big]
\\[4pt]
&\hspace{3cm}\times
\Big[
\omega_2 x^{\omega_2}
+\left(3\eta_2-4\omega_2\eta_2\ln x\right)x^{2\omega_2}
\Big]\,dx.
\end{aligned}
\]
Expanding the integrand and integrating term by term yields integrals
of the form
\[
\int_0^1 x^a\,dx=\frac{1}{a+1},\qquad
\int_0^1 x^a\ln x\,dx=-\frac{1}{(a+1)^2},\qquad
\int_0^1 x^a(\ln x)^2\,dx=\frac{2}{(a+1)^3}.
\]

After simplification, the stress--strength reliability is obtained as
\begin{align}\label{uss}
R&=\frac{\omega_1}{(\omega_1+3\eta_1)(\omega_2+3\eta_2)}\Bigg[
\frac{\omega_1\omega_2}{\omega_1+\omega_2}
+\omega_2\eta_1\Bigg(\frac{2}{2\omega_1+\omega_2}+\frac{8\omega_1}{(2\omega_1+\omega_2)^2}\Bigg) \nonumber \\
&\quad+\omega_1\eta_2\Bigg(\frac{3}{\omega_1+2\omega_2}+\frac{4\omega_2}{(\omega_1+2\omega_2)^2}\Bigg)+\eta_1\eta_2\Bigg(
\frac{6}{2\omega_1+2\omega_2}+\frac{24\omega_1+8\omega_2}{(2\omega_1+2\omega_2)^2}+\frac{64\omega_1\omega_2}{(2\omega_1+2\omega_2)^3}
\Bigg)
\Bigg].
\end{align}
It is worth noting that the random variables $X_i \sim USh(\omega_i, \eta_i)$ are derived from the transformation $X_i=e^{-Y_i}$, where $Y_i \sim Sh(\omega_i, \eta_i)$. Since the transformation $x\mapsto e^{-x}$ is strictly
monotone decreasing, it follows that
\[
P(X_2 < X_1)= P(Y_2 > Y_1)=1-P(Y_2 < Y_1).
\]
Consequently, the stress--strength reliability of the UShD is the failure probability of ShD, i.e., $R_{\mathrm{USh}}=1-R_{\mathrm{Sh}}$.
We now investigate some special cases of \eqref{uss}.
\begin{enumerate}
  \item Setting \(\eta_1 = \eta_2 = 0\), we obtain the stress--strength reliability of two independent $\text{Beta}(\omega_1, 1)$ and $\text{Beta}(\omega_2, 1)$ distributed random variables: $R=\frac{\omega_1}{{\omega_1+\omega_2}}$, which also coincides with the failure probability of two independent  $\text{Exp}(\omega_1)$ and $ \text{Exp}(\omega_2)$ random variables.
  \item Setting \(\eta_1 = \eta_2 = \eta\) and \(\omega_1 = \omega_2 = \omega\), i.e., $X_1\sim USh(\omega, \eta)$ and $X_2 \sim USh(\omega, \eta)$, we obtain $R=\frac{1}{2}$.
\item Setting \(\omega_1 = \omega_2 = \omega\) and \(\eta_2 = 0\),  i.e., $X_1 \sim USh(\omega, \eta)$, $X_2 \sim \text{Beta}(\omega, 1)$, we obtain
$ R = \frac{9 \omega+28 \eta}{18(\omega+3\eta)}$.
\end{enumerate}

\section{Statistical Properties}
This section discusses the main statistical properties of the proposed UShD.
\subsection{Moments and Related Measures}
\begin{theorem}
The moment generating function of the unit Shiha distributed random variable is given by
\begin{equation}\label{eq:mgf}
M_X(t)=\frac{\omega}{\omega+3\eta}\sum_{\kappa=0}^{\infty}\frac{t^{\kappa}}{{\kappa}!}\left[\frac{\omega}{{\kappa}+\omega}
+\frac{2\eta}{{\kappa}+2\omega}+\frac{8\omega\eta}{({\kappa}+2\omega)^2}\right].
\end{equation}
\end{theorem}
\begin{proof} By definition, $M_X(t)=E[e^{tX}]$. Using the Taylor expansion $e^{tX}=\sum_{\kappa=0}^\infty t^{\kappa} X^{\kappa}/\kappa!$ and interchanging sum and expectation,
\begin{equation}
M_X(t)= \sum_{{\kappa}=0}^{\infty}\frac{t^{\kappa}}{{\kappa}!}\,E[X^{\kappa}],
\end{equation}
and the $\kappa$-th moment is given by
\begin{align}\label{eq:moments}
\mu'_{\kappa} =E[X^{\kappa}] &=\int_0^1 x^{\kappa} f(x;\omega,\eta)\,dx \nonumber \\
&= \frac{\omega}{\omega+3\eta}\Big[\omega\!\int_0^1 x^{{\kappa}+\omega-1}dx + 2\eta\!\int_0^1 x^{{\kappa}+2\omega-1}dx - 8\omega\eta\!\int_0^1 (\ln x) x^{{\kappa}+2\omega-1}dx\Big] \nonumber \\
&=\frac{\omega}{\omega+3\eta}\Big[\frac{\omega}{{\kappa}+\omega}+\frac{2\eta}{{\kappa}+2\omega}+ \frac{8\omega\eta}{({\kappa}+2\omega)^2}\Big].
\end{align}
\end{proof}
The first four moments of the UShD can be obtained by setting $\kappa = 1,2,3$, and $4$ in (\ref{eq:moments}),
\begin{align}
\mu'_1 =E[X] &= \frac{\omega}{\omega+3\eta}\left[\frac{\omega}{1+\omega}+\frac{2\eta}{1+2\omega}+\frac{8\omega\eta}{(1+2\omega)^2}\right], \label{eq:mean}\\
\mu'_2 =E[X^2] &= \frac{\omega}{\omega+3\eta}\left[\frac{\omega}{2+\omega}+\frac{\eta}{1+\omega}+\frac{2\omega\eta}{(1+\omega)^2}\right]. \label{eq:second}\\
\mu'_3 =E[X^3]& =\frac{\omega}{\omega+3\eta}\Big[\frac{\omega}{3+\omega}+\frac{2\eta}{3+2\omega}+ \frac{8\omega\eta}{(3+2\omega)^2}\Big] \\
\mu'_4 =E[X^4] &= \frac{\omega}{\omega+3\eta}\Big[\frac{\omega}{4+\omega}+\frac{\eta}{2+\omega}+ \frac{2\omega\eta}{(2+\omega)^2}\Big].
\end{align}
 The numerical values of the first four moments of UShD for several values of $\eta$ and $\omega$ are presented in Table~\ref{T:mom}. It can be observed that the first four moments increase as $\omega$ increases and decrease as $\eta$ increases.
\begin{table}[!h]
\centering
\caption{The first four moments of the UShD for different values of the model parameters.}\label{T:mom}
\adjustbox{max width=\textwidth, center}{\footnotesize
\begin{tabular}{c|cccc|cccc|cccc} \hline
\multirow{2}{*}{$ \omega$}
& \multicolumn{4}{c|}{$\eta=0.1$}
& \multicolumn{4}{c|}{$\eta=0.8$}
& \multicolumn{4}{c}{$\eta=1.5$} \\ \cline{2-13}
 & $\mu'_1$ & $\mu'_2$ & $\mu'_3$ & $\mu'_4$ & $\mu'_1$ & $\mu'_2$ & $\mu'_3$ & $\mu'_4$ & $\mu'_1$ & $\mu'_2$ & $\mu'_3$ & $\mu'_4$ \\ \hline
0.1 & 0.078 & 0.039 & 0.026 & 0.019 & 0.075 & 0.036 & 0.024 & 0.018 & 0.074 & 0.036 & 0.024 & 0.018 \\
0.2 & 0.156 & 0.081 & 0.054 & 0.041 & 0.151 & 0.075 & 0.05 & 0.037 & 0.15 & 0.075 & 0.049 & 0.036 \\
0.3 & 0.225 & 0.121 & 0.082 & 0.062 & 0.22 & 0.114 & 0.076 & 0.056 & 0.22 & 0.114 & 0.075 & 0.056 \\
0.4 & 0.283 & 0.159 & 0.11 & 0.084 & 0.281 & 0.152 & 0.102 & 0.076 & 0.28 & 0.151 & 0.101 & 0.075 \\
0.5 & 0.333 & 0.194 & 0.136 & 0.104 & 0.333 & 0.188 & 0.128 & 0.096 & 0.333 & 0.187 & 0.127 & 0.095 \\
0.6 & 0.377 & 0.227 & 0.161 & 0.124 & 0.379 & 0.221 & 0.153 & 0.116 & 0.38 & 0.22 & 0.152 & 0.115 \\
0.7 & 0.415 & 0.257 & 0.184 & 0.144 & 0.419 & 0.252 & 0.177 & 0.135 & 0.42 & 0.251 & 0.176 & 0.134 \\
0.8 & 0.448 & 0.284 & 0.207 & 0.162 & 0.454 & 0.281 & 0.2 & 0.154 & 0.456 & 0.281 & 0.199 & 0.152 \\
0.9 & 0.478 & 0.31 & 0.228 & 0.18 & 0.485 & 0.308 & 0.222 & 0.172 & 0.487 & 0.308 & 0.221 & 0.17 \\
1 & 0.504 & 0.333 & 0.248 & 0.197 & 0.513 & 0.333 & 0.243 & 0.19 & 0.515 & 0.333 & 0.242 & 0.188 \\
1.1 & 0.528 & 0.355 & 0.267 & 0.213 & 0.538 & 0.357 & 0.263 & 0.206 & 0.54 & 0.357 & 0.262 & 0.205 \\
1.2 & 0.55 & 0.376 & 0.285 & 0.228 & 0.56 & 0.378 & 0.282 & 0.223 & 0.563 & 0.379 & 0.281 & 0.221 \\
1.3 & 0.57 & 0.395 & 0.302 & 0.243 & 0.58 & 0.399 & 0.3 & 0.238 & 0.583 & 0.4 & 0.299 & 0.237 \\
1.4 & 0.588 & 0.413 & 0.318 & 0.258 & 0.599 & 0.418 & 0.317 & 0.254 & 0.602 & 0.419 & 0.317 & 0.252 \\
1.5 & 0.604 & 0.43 & 0.333 & 0.271 & 0.615 & 0.436 & 0.333 & 0.268 & 0.619 & 0.437 & 0.333 & 0.267 \\
1.6 & 0.619 & 0.447 & 0.348 & 0.285 & 0.631 & 0.452 & 0.349 & 0.282 & 0.634 & 0.454 & 0.349 & 0.281 \\
1.7 & 0.634 & 0.462 & 0.362 & 0.298 & 0.645 & 0.468 & 0.364 & 0.296 & 0.649 & 0.47 & 0.364 & 0.295 \\
1.8 & 0.647 & 0.476 & 0.376 & 0.31 & 0.658 & 0.483 & 0.378 & 0.309 & 0.662 & 0.485 & 0.379 & 0.308 \\
1.9 & 0.659 & 0.49 & 0.389 & 0.322 & 0.67 & 0.497 & 0.391 & 0.321 & 0.674 & 0.499 & 0.392 & 0.321 \\
2 & 0.67 & 0.502 & 0.401 & 0.333 & 0.681 & 0.51 & 0.404 & 0.333 & 0.685 & 0.513 & 0.406 & 0.333 \\
2.1 & 0.681 & 0.515 & 0.413 & 0.344 & 0.692 & 0.523 & 0.417 & 0.345 & 0.696 & 0.526 & 0.418 & 0.345 \\
2.2 & 0.691 & 0.526 & 0.424 & 0.355 & 0.701 & 0.534 & 0.429 & 0.356 & 0.705 & 0.538 & 0.43 & 0.357 \\
2.3 & 0.7 & 0.537 & 0.435 & 0.366 & 0.711 & 0.546 & 0.44 & 0.367 & 0.715 & 0.549 & 0.442 & 0.368 \\
2.4 & 0.709 & 0.548 & 0.446 & 0.376 & 0.719 & 0.556 & 0.451 & 0.378 & 0.723 & 0.56 & 0.453 & 0.378 \\
2.5 & 0.717 & 0.558 & 0.456 & 0.385 & 0.727 & 0.567 & 0.462 & 0.388 & 0.731 & 0.57 & 0.464 & 0.389 \\
2.6 & 0.725 & 0.568 & 0.466 & 0.395 & 0.735 & 0.576 & 0.472 & 0.398 & 0.739 & 0.58 & 0.474 & 0.399 \\
2.7 & 0.732 & 0.577 & 0.475 & 0.404 & 0.742 & 0.586 & 0.481 & 0.407 & 0.746 & 0.589 & 0.484 & 0.408 \\
2.8 & 0.739 & 0.586 & 0.484 & 0.413 & 0.749 & 0.595 & 0.491 & 0.416 & 0.753 & 0.598 & 0.493 & 0.418 \\
2.9 & 0.746 & 0.594 & 0.493 & 0.421 & 0.755 & 0.603 & 0.5 & 0.425 & 0.759 & 0.607 & 0.502 & 0.427 \\
3 & 0.752 & 0.602 & 0.502 & 0.43 & 0.761 & 0.611 & 0.508 & 0.434 & 0.765 & 0.615 & 0.511 & 0.435 \\
\hline
\end{tabular}}
\end{table}
Using the first four moments, coefficient of skewness (SK) and coefficient of kurtosis (KU) are computed as:
\begin{equation*}
\text{SK}=\frac{\mu'_3 - 3\mu'_2 \, \mu'_1 + 2(\mu'_1)^3}{{(\text{Var}(Y))}^{3/2}},\qquad \qquad
\text{KU}=\frac{\mu'_4 - 4\mu'_3 \, \mu'_1 + 6\mu'_2 \, (\mu'_1)^2 - 3{\mu'_1}^4}{{(\mathrm{Var}(Y))}^{2}}.
\end{equation*}
Table~\ref{T:mom2} reports the variance, skewness, and kurtosis values of the UShD for various combinations of $\omega$ and $\eta$.
From Table~\ref{T:mom2}, it is observed that
\begin{enumerate}
  \item For fixed $\omega$, the variance decreases as $\eta$ increases, while skewness increases if $\omega \in (0,1)$ .
  \item For fixed $\eta$,
  \begin{itemize}
    \item the variance initially increases for small $\omega$ values and then decreases.
    \item the skewness decreases as $\omega$ increases, being positive for $\omega \in (0,1)$ and negative for $\omega > 1$,
    \item the kurtosis decreases when $0< \omega < 1$ and then increases for $\omega \geq 1$.
  \end{itemize}
  Figure~\ref{fig:varsk} further illustrates the effects of the parameters. These results indicate that the parameter $\omega$ primarily governs the direction of skewness of the distribution, producing right-skewed distributions for $\omega < 1$, left-skewed distributions for $\omega > 1$, and near to be symmetric at $\omega=1$. On the other hand, the parameter $\eta$ controls the level of dispersion and tail behavior, where larger values of $\eta$ lead to more concentrated distributions with heavier tails. this structural flexibility makes the UShD well suited for modeling different types of real-life data with varying degrees of asymmetry and tail heaviness.
\end{enumerate}
\begin{table}[!ht]
\centering
\caption{The Variance, skewness, and kurtosis of the UShD for different values of $\omega$ and $\eta$.} \label{T:mom2}
\adjustbox{max width=\textwidth, center}{
\begin{tabular}{c|ccc|ccc|ccc|ccc}
\hline
\multirow{2}{*}{$\omega$}
& \multicolumn{3}{c|}{$\eta=0.1$}
& \multicolumn{3}{c|}{$\eta=0.8$}
& \multicolumn{3}{c|}{$\eta=1.5$}
& \multicolumn{3}{c}{$\eta=3$} \\ \cline{2-13}
& $Var(Y)$ & SK & KU
& $Var(Y)$ &SK & KU
& $Var(Y)$ & SK & KU
& $Var(Y)$ &SK & KU \\ \hline
0.1 & 0.033 & 2.968 & 11.66 & 0.031 & 3.066 & 12.389 & 0.031 & 3.075 & 12.456 & 0.03 & 3.08 & 12.495 \\
0.2 & 0.056 & 1.78 & 5.285 & 0.053 & 1.847 & 5.63 & 0.052 & 1.854 & 5.668 & 0.052 & 1.859 & 5.691 \\
0.3 & 0.071 & 1.235 & 3.436 & 0.066 & 1.281 & 3.638 & 0.065 & 1.287 & 3.663 & 0.065 & 1.29 & 3.68 \\
0.4 & 0.079 & 0.897 & 2.636 & 0.073 & 0.929 & 2.775 & 0.072 & 0.934 & 2.796 & 0.072 & 0.937 & 2.809 \\
0.5 & 0.083 & 0.656 & 2.23 & 0.077 & 0.679 & 2.342 & 0.076 & 0.683 & 2.36 & 0.075 & 0.686 & 2.372 \\
0.6 & 0.085 & 0.471 & 2.014 & 0.077 & 0.488 & 2.112 & 0.076 & 0.492 & 2.129 & 0.075 & 0.495 & 2.141 \\
0.7 & 0.085 & 0.322 & 1.9 & 0.077 & 0.334 & 1.993 & 0.075 & 0.338 & 2.011 & 0.074 & 0.341 & 2.023 \\
0.8 & 0.083 & 0.196 & 1.848 & 0.075 & 0.206 & 1.94 & 0.073 & 0.21 & 1.958 & 0.072 & 0.213 & 1.971 \\
0.9 & 0.081 & 0.089 & 1.836 & 0.073 & 0.096 & 1.928 & 0.071 & 0.1 & 1.947 & 0.069 & 0.105 & 1.96 \\
1 & 0.079 & -0.005 & 1.849 & 0.07 & 0 & 1.944 & 0.068 & 0.005 & 1.964 & 0.066 & 0.01 & 1.976 \\
1.1 & 0.076 & -0.088 & 1.881 & 0.067 & -0.085 & 1.98 & 0.065 & -0.079 & 1.999 & 0.063 & -0.073 & 2.011 \\
1.2 & 0.074 & -0.163 & 1.925 & 0.065 & -0.162 & 2.028 & 0.062 & -0.155 & 2.048 & 0.06 & -0.147 & 2.059 \\
1.3 & 0.071 & -0.231 & 1.979 & 0.062 & -0.231 & 2.086 & 0.06 & -0.223 & 2.105 & 0.058 & -0.214 & 2.116 \\
1.4 & 0.068 & -0.292 & 2.038 & 0.059 & -0.293 & 2.15 & 0.057 & -0.285 & 2.17 & 0.055 & -0.275 & 2.179 \\
1.5 & 0.065 & -0.349 & 2.103 & 0.057 & -0.351 & 2.219 & 0.054 & -0.343 & 2.239 & 0.052 & -0.331 & 2.247 \\
1.6 & 0.063 & -0.401 & 2.17 & 0.055 & -0.405 & 2.292 & 0.052 & -0.396 & 2.312 & 0.05 & -0.383 & 2.317 \\
1.7 & 0.06 & -0.449 & 2.24 & 0.052 & -0.454 & 2.366 & 0.05 & -0.445 & 2.387 & 0.047 & -0.431 & 2.391 \\
1.8 & 0.058 & -0.494 & 2.311 & 0.05 & -0.501 & 2.443 & 0.047 & -0.491 & 2.464 & 0.045 & -0.475 & 2.465 \\
1.9 & 0.056 & -0.537 & 2.383 & 0.048 & -0.544 & 2.52 & 0.045 & -0.534 & 2.542 & 0.043 & -0.517 & 2.541 \\
2 & 0.053 & -0.576 & 2.456 & 0.046 & -0.584 & 2.597 & 0.043 & -0.574 & 2.62 & 0.041 & -0.556 & 2.617 \\
2.1 & 0.051 & -0.613 & 2.528 & 0.044 & -0.623 & 2.675 & 0.042 & -0.612 & 2.698 & 0.039 & -0.594 & 2.694 \\
2.2 & 0.049 & -0.648 & 2.6 & 0.042 & -0.659 & 2.752 & 0.04 & -0.648 & 2.776 & 0.038 & -0.629 & 2.77 \\
2.3 & 0.047 & -0.681 & 2.672 & 0.041 & -0.693 & 2.829 & 0.038 & -0.682 & 2.854 & 0.036 & -0.662 & 2.846 \\
2.4 & 0.045 & -0.712 & 2.744 & 0.039 & -0.725 & 2.905 & 0.037 & -0.714 & 2.932 & 0.035 & -0.693 & 2.922 \\
2.5 & 0.044 & -0.742 & 2.814 & 0.038 & -0.756 & 2.98 & 0.035 & -0.745 & 3.008 & 0.033 & -0.723 & 2.998 \\
2.6 & 0.042 & -0.77 & 2.884 & 0.036 & -0.785 & 3.054 & 0.034 & -0.774 & 3.084 & 0.032 & -0.752 & 3.072 \\
2.7 & 0.041 & -0.797 & 2.953 & 0.035 & -0.812 & 3.127 & 0.033 & -0.802 & 3.159 & 0.031 & -0.779 & 3.146 \\
2.8 & 0.039 & -0.823 & 3.021 & 0.034 & -0.839 & 3.199 & 0.032 & -0.829 & 3.233 & 0.029 & -0.805 & 3.219 \\
2.9 & 0.038 & -0.847 & 3.088 & 0.033 & -0.864 & 3.27 & 0.03 & -0.854 & 3.306 & 0.028 & -0.83 & 3.292 \\
3 & 0.036 & -0.871 & 3.154 & 0.031 & -0.888 & 3.34 & 0.029 & -0.879 & 3.378 & 0.027 & -0.854 & 3.363 \\
 \hline
\end{tabular}}
\end{table}
\begin{figure}[H]
  \centering
 \includegraphics[width=15cm, height=3.9cm]{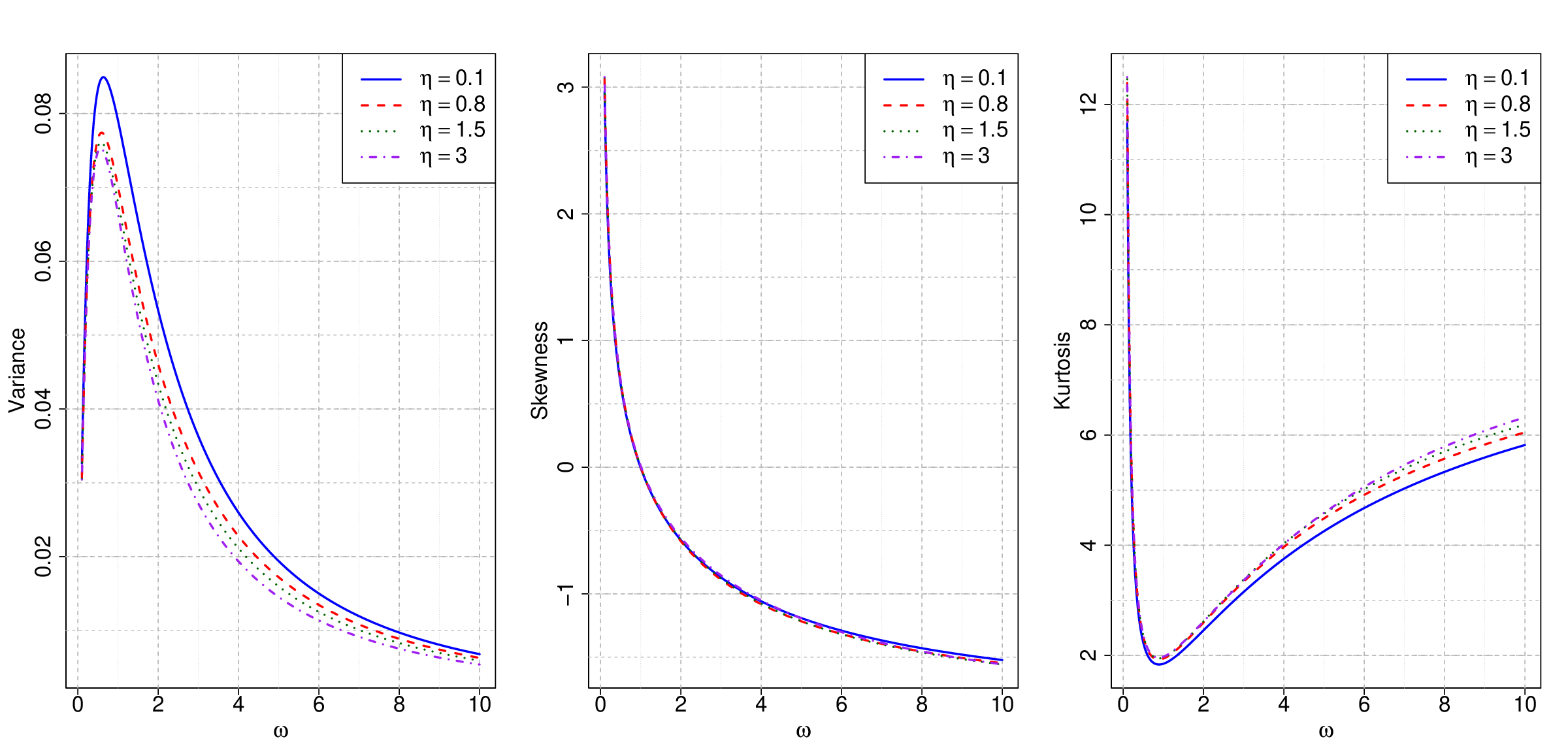}
 \includegraphics[width=15cm, height=6cm]{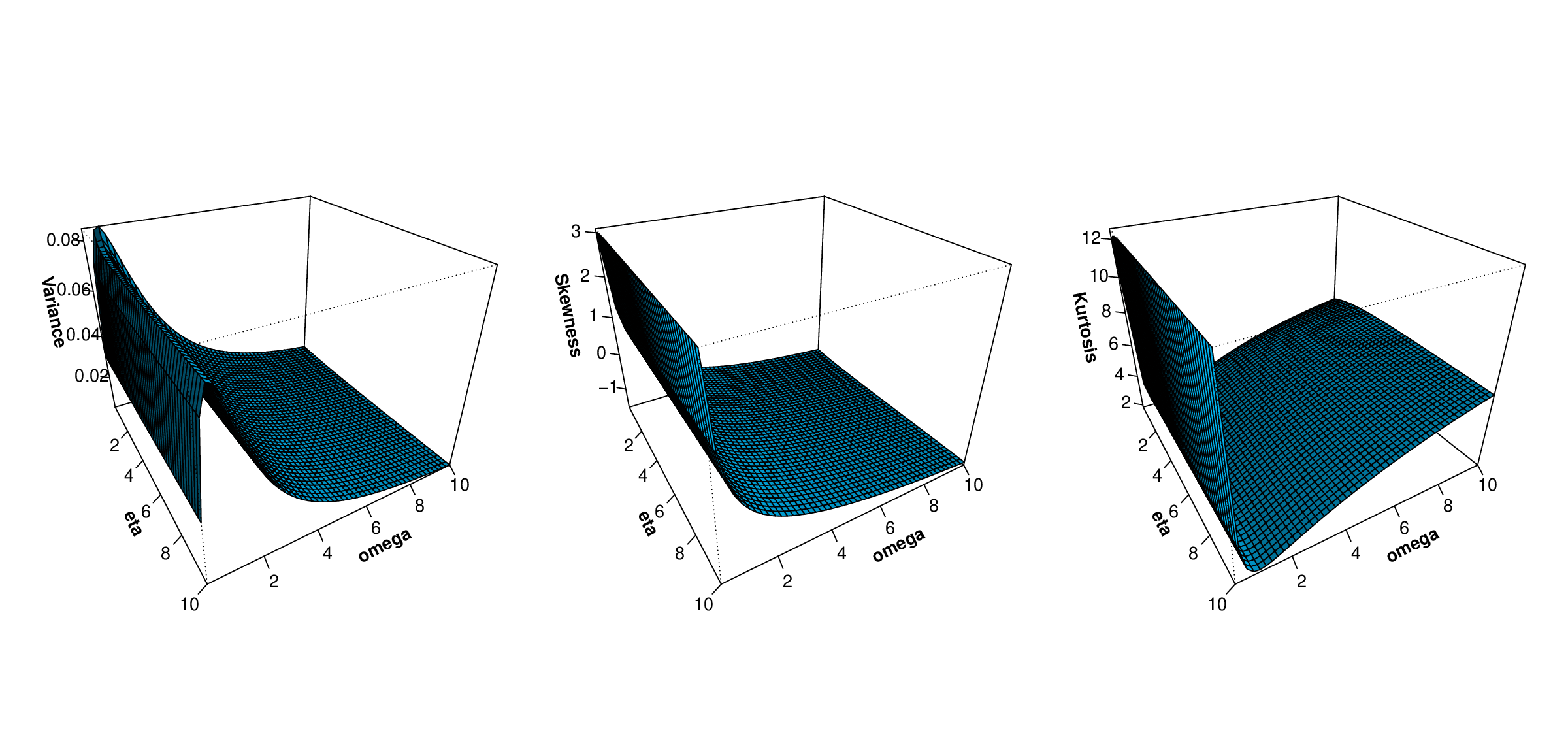}
\caption{The variance, skewness, and kurtosis of the UShD for $0< \omega, \eta \leq 10$.}
 \label{fig:varsk}
\end{figure}

\subsection{Quantile Function}
The quantile function $x_p = F^{-1}(p)$, for $p \in (0,1)$, of the UShD is the solution of
\begin{equation} \label{eq:qu1}
 \omega x_p^{\omega}+(3\eta-4\omega\eta\ln x_p)x_p^{2\omega}=p(\omega+3\eta), \quad p \in (0,1)
\end{equation}

Equation~(\ref{eq:qu1}) does not admit a closed-form solution for $x_p$. Therefore, the quantiles are computed numerically usingthe \texttt{uniroot}
function in the R statistical software.
Table~\ref{tab:quantiles} reports selected quantile values for different combinations of $\omega$ and $\eta$. It can be observed that, for a fixed value of $\eta$, the quantiles increase as $\omega$ increases. For a fixed $\omega$, increasing $\eta$ leads to higher quantile values for probabilities $p \in (0, 0.57]$, whereas for $p \in (0.57, 1)$ the quantiles decrease slightly as $\eta$ increases.

\begin{table}[h!]
\centering \small{
\caption{Quantile values of UShD for different parameter combinations.}
\label{tab:quantiles}
\renewcommand{\arraystretch}{1.3}

\begin{tabular}{|c|c|c|c|c|c|c|}
\hline
\textbf{p} &
\begin{tabular}{c}
$\omega=0.5$ \\
$\eta=0.4$
\end{tabular} &
\begin{tabular}{c}
$\omega=0.5$ \\
$\eta=1$
\end{tabular} &
\begin{tabular}{c}
$\omega=1$ \\
$\eta=0.4$
\end{tabular} &
\begin{tabular}{c}
$\omega=1$ \\
$\eta=1$
\end{tabular} &
\begin{tabular}{c}
$\omega=1.5$ \\
$\eta=0.4$
\end{tabular} &
\begin{tabular}{c}
$\omega=1.5$ \\
$\eta=1$
\end{tabular} \\
\hline

0.01 & 0.0006 & 0.0011 & 0.0192 & 0.0272 & 0.0650 & 0.0818 \\
0.05 & 0.0080 & 0.0102 & 0.0779 & 0.0926 & 0.1723 & 0.1954 \\
0.10 & 0.0224 & 0.0262 & 0.1374 & 0.1533 & 0.2564 & 0.2782 \\
0.25 & 0.0880 & 0.0933 & 0.2869 & 0.2992 & 0.4284 & 0.4424 \\
0.40 & 0.1821 & 0.1861 & 0.4214 & 0.4281 & 0.5589 & 0.5657 \\
0.50 & 0.2615 & 0.2635 & 0.5091 & 0.5119 & 0.6363 & 0.6390 \\
0.55 & 0.3067 & 0.3075 & 0.5531 & 0.5540 & 0.6734 & 0.6743 \\
0.57 & 0.3260 & 0.3262 & 0.5708 & 0.5710 & 0.6880 & 0.6882\\
0.58 & 0.3359 & 0.3358 & 0.5797 & 0.5795 & 0.6952 & 0.6952 \\
0.60 & 0.3562 & 0.3555 & 0.5975 & 0.5967 & 0.7097 & 0.7090 \\
0.75 & 0.5347 & 0.5292 & 0.7353 & 0.7301 & 0.8166 & 0.8124 \\
0.90 & 0.7770 & 0.7693 & 0.8860 & 0.8802 & 0.9244 & 0.9201 \\
0.95 & 0.8800 & 0.8741 & 0.9411 & 0.9372 & 0.9616 & 0.9588 \\
0.99 & 0.9743 & 0.9726 & 0.9879 & 0.9868 & 0.9922 & 0.9915 \\
\hline
\end{tabular}}
\end{table}

\begin{figure}[h]
  \centering
 \includegraphics[width=14cm, height=5cm]{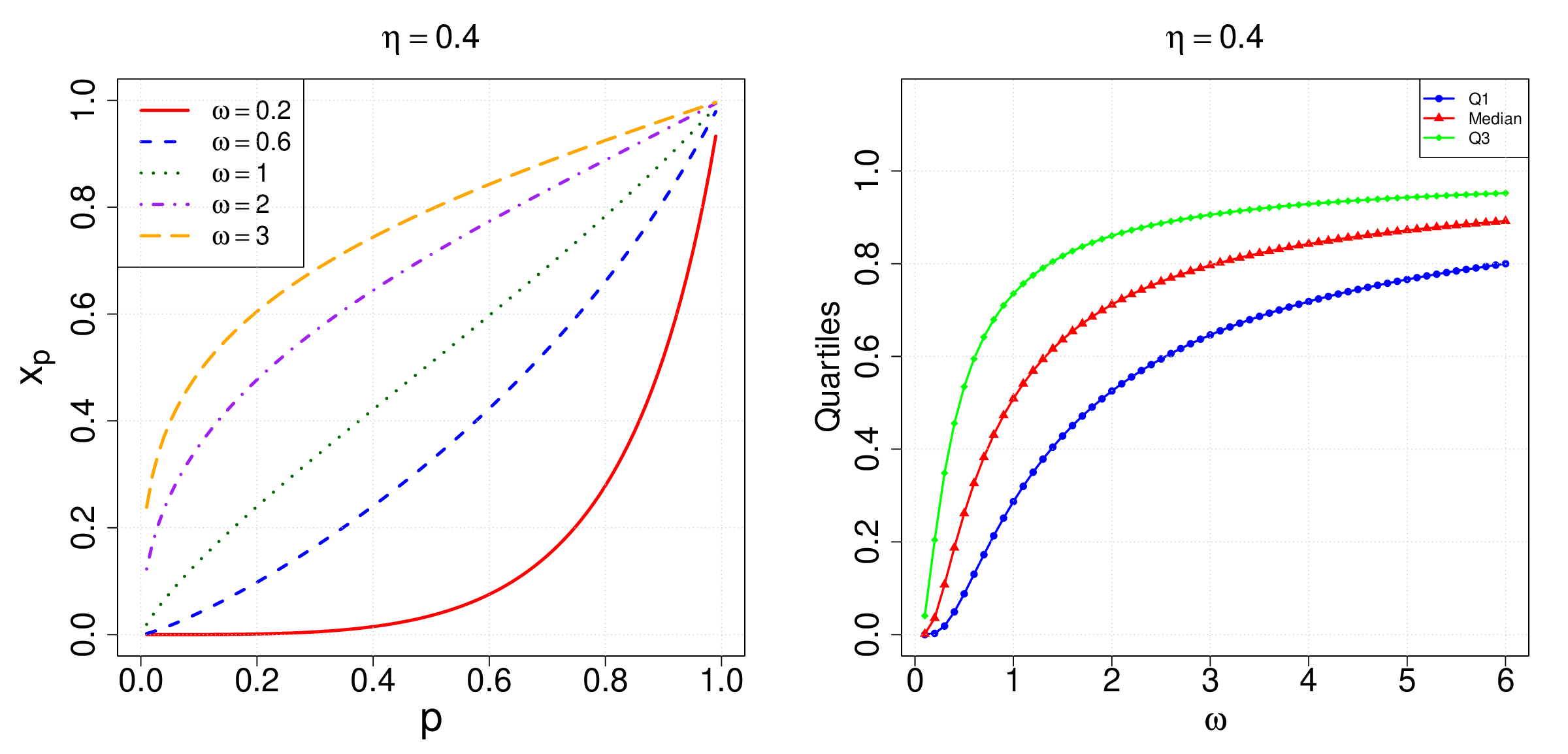}

  \caption{Left: quantile functions of UShD for $\eta = 0.4$ and various $\omega$ values.
  Right: quartiles of UShD as functions of $\omega$ for $\eta = 0.4$.}
 \label{fig:qua}
\end{figure}

Figure~\ref{fig:qua} illustrates the effect of the shape parameter $\omega$ on the quantiles of the UShD for a fixed value of $\eta = 0.4$.
 The left panel shows the entire quantile functions across probability levels, while the right panel displays the variation of the first quartile, median, and third quartile as functions of $\omega$, all quartiles increase monotonically with increasing values of $\omega$.
\subsection{Entropy}
Entropy is a fundamental measure of uncertainty or information associated with a random variable.
For continuous distributions, the differential entropy quantifies the unpredictability of the outcomes and is widely used in statistics, information theory, and engineering for model comparison, risk assessment, and system analysis.
For a random variable $X\sim \text{USh}(\omega,\eta)$, the differential entropy is defined as
\[
H(x) = -\int_{0}^{1} f(x;\omega,\eta)\,\ln\!\bigl(f(x;\omega,\eta)\bigr)\,dx= -\text{E}\bigl[\ln f(X;\omega,\eta)\bigr].
\]
Substituting the PDF from (\ref{pdf}), we obtain
\begin{align}
H(x) &= -\text{E}\left[\ln\!\left(\frac{\omega}{\omega+3\eta}\left(\omega X^{\omega-1}+\left(2\eta-8\omega\eta\ln X\right)X^{2\omega-1}\right)\right)
\right] \nonumber \\
&= -\ln\!\left(\frac{\omega}{\omega+3\eta}\right)-\text{E}\left[\ln\!\left(\omega X^{\omega-1}+(2\eta-8\omega\eta\ln X)X^{2\omega-1}
\right)\right]. \label{eq:ent}
\end{align}
The expected value in (\ref{eq:ent})  cannot be derived analytically, but it can be evaluated numerically for specific parameter values, as illustrated in Table \ref{T:entropy} and Figure~\ref{fig:entropy}. It can be observed that $H(X)$ is negative and increases with $\omega$ for $0 < \omega < 1$, and then decreases slowly for $\omega \geq 1$.

\begin{table}[h!]
\centering
\caption{Entropy of UShD for selected parameters values.}
\label{T:entropy}
\adjustbox{max width=\textwidth, center}{
\begin{tabular}{c|ccccccccccc}
\hline
$\omega \backslash \eta$ & 0.2 & 0.5 & 0.7 & 0.9 &1 &1.1 &1.5 &2 &3 &4 \\
\hline
   0.1  & 	-3.4946 &   	-3.5412  &  	-3.5511&	-3.5568 &	-3.5588   &	-3.5605   &	-3.5649&	-3.5680 &	-3.5712 &	-3.5728 \\
   0.2  &	-1.7715   &	-1.7418 &   	-1.7350&	-1.7311 &	-1.7297  &  	-1.7286   &	-1.7254&	-1.7232 &	-1.7210 &	-1.7198\\
  0.3 &   	-0.8691  &  	-0.8240 &   	-0.8132&	-0.8069 &	-0.8046 &   	-0.8027 &   	-0.7975&	-0.7938 &	-0.7901 &	-0.7882\\
  0.5  &  	-0.2268 &   	-0.2040  &  	-0.1984&	-0.1951 &	-0.1939&    	-0.1930  &  	-0.1904&	-0.1886 &	-0.1868&  	-0.1859\\
   0.7  &  	-0.0469  &  	-0.0414 &   	-0.0407&	-0.0405 &	-0.0405&    	-0.0405   &	-0.0408&	-0.0411 &	-0.0417 &	-0.0420\\
  0.9 &   	-0.0043 &   	-0.0092  &  	-0.0120&	-0.0142&  	-0.0151  &  	-0.0159 &   	-0.0185&	-0.0207 &	-0.0233 &	-0.0248\\
  1.0  &  	-0.0052  &  	-0.0138  &  	-0.0179&	-0.0210 &	-0.0223   &	-0.0235  &  	-0.0271&	-0.0301&  	-0.0336 &	-0.0356\\
 1.1  &  	-0.0143&    	-0.0259  &  	-0.0311&	-0.0350 &	-0.0367   &	-0.0381 &   	-0.0426&	-0.0463 &	-0.0507&  	-0.0533\\
 1.5&    	-0.0910  &  	-0.1093&    	-0.1176&	-0.1239 &	-0.1265&    	-0.1288&    	-0.1361&	-0.1423 &	-0.1498 &	-0.1542\\
 2.0&    	-0.2143  &  	-0.2357  &  	-0.2457&	-0.2536 &	-0.2570 &   	-0.2600 &   	-0.2696&	-0.2780&  	-0.2884 &	-0.2946\\
  2.5 &   	-0.3370 &   	-0.3590  &  	-0.3697&	-0.3784&  	-0.3821&    	-0.3855 &   	-0.3966&	-0.4065 &	-0.4191&  	-0.4268\\
3.0   &	-0.4515  &  	-0.4730 &   	-0.4840&	-0.4930 &	-0.4969 &   	-0.5005  &  	-0.5124&	-0.5233&  	-0.5376&  	-0.5466\\

\hline
\end{tabular}}
\end{table}

\begin{figure}[htbp]
  \centering
 \includegraphics[width=14cm, height=5cm]{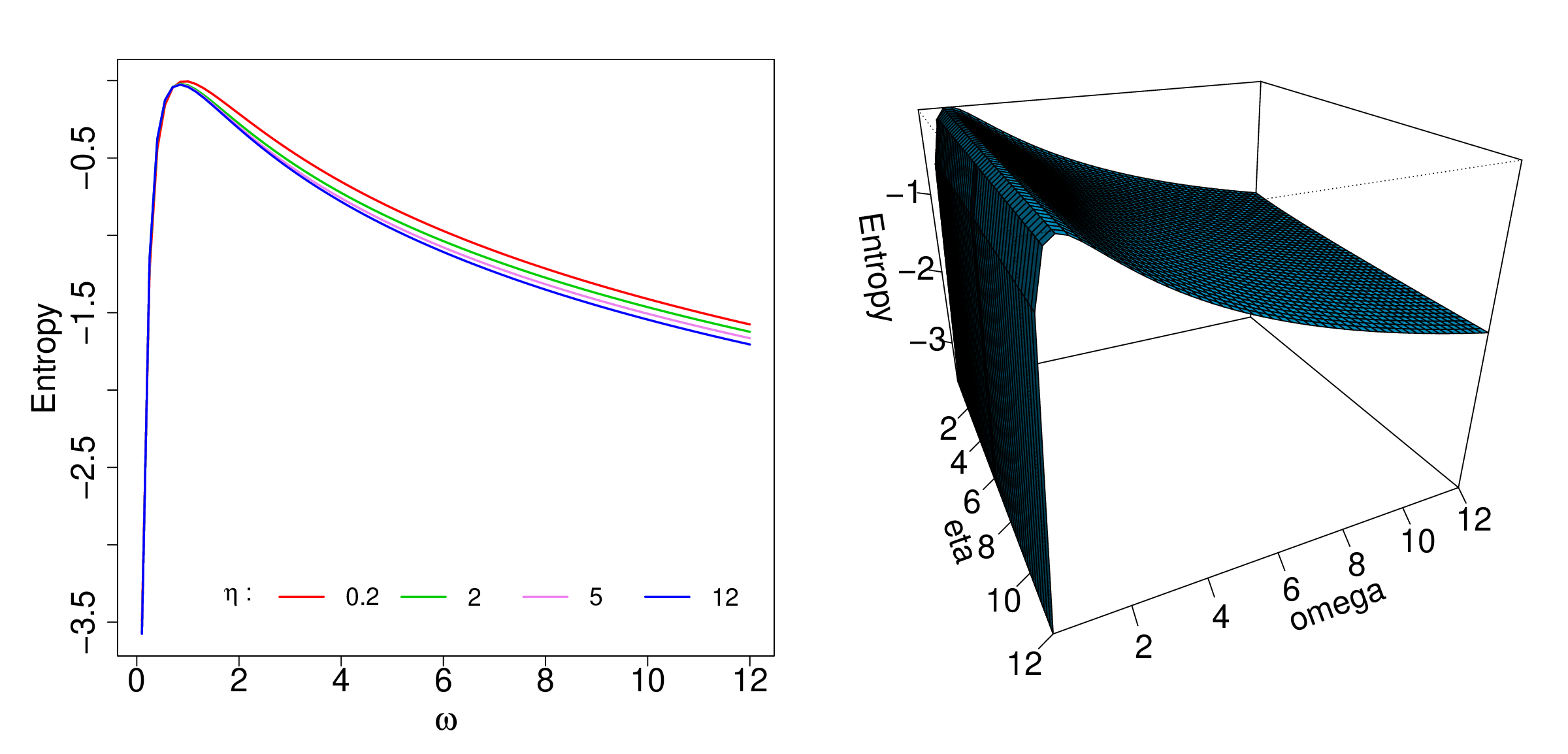}
 \caption{The entropy of UShD for $0< \omega, \eta \leq 12$, the entropy range is [-3.576, -0.003]}
 \label{fig:entropy}
\end{figure}

\section{Parameter Estimation}
In this section, we use the maximum likelihood estimation method to estimate the parameters $\omega$ and $\eta$ of the UShD.
Let $X_1,\ldots,X_n$ be a random sample from the UShD and $x_1,\ldots,x_n$ represents the values of the sample.
The likelihood and log-likelihood functions are given by, respectively
\begin{equation*}
L(\omega,\eta)
=
\prod_{i=1}^n
\frac{\omega}{\omega+3\eta}\left[\omega x_i^{\omega-1}+(2\eta-8\omega\eta\ln x_i)x_i^{2\omega-1}\right].
\end{equation*}
\begin{equation*}
\ell(\omega,\eta)=n\ln\omega-n\ln(\omega+3\eta)+\sum_{i=1}^n\ln \left[ \omega x_i^{\omega-1}+(2\eta-8\omega\eta\ln x_i)x_i^{2\omega-1}\right].
\end{equation*}
The first partial derivatives of $\ell(\omega,\eta) $ with respect to $\omega$ and $\eta$ are
\begin{equation*}
\frac{\partial \ell}{\partial \omega}=\frac{n}{\omega}-\frac{n}{\omega+3\eta}+\sum_{i=1}^n
\frac{x_i^{\omega-1}\bigl(1+\omega\ln x_i\bigr)+x_i^{2\omega-1}\left[
(2\eta-8\omega\eta\ln x_i)(2\ln x_i)-8\eta\ln x_i\right]}{\omega x_i^{\omega-1}+(2\eta-8\omega\eta\ln x_i)x_i^{2\omega-1}}.
\end{equation*}

\begin{equation*}
\frac{\partial \ell}{\partial \eta}=-\frac{3n}{\omega+3\eta}+\sum_{i=1}^n\frac{(2-8\omega\ln x_i)x_i^{2\omega-1}}{\omega x_i^{\omega-1}+(2\eta-8\omega\eta\ln x_i)x_i^{2\omega-1}}.
\end{equation*}
 The maximum likelihood estimates (MLEs) of the parameters $\omega$ and
$\eta$ are obtained by solving the system of likelihood equations
\[
\frac{\partial \ell}{\partial \omega}=0
\quad \text{and} \quad
\frac{\partial \ell}{\partial \eta}=0.
\]
However, due to the complexity of these equations, closed-form solutions
are not available. Consequently, numerical optimization methods, such as
the Limited-memory Broyden--Fletcher--Goldfarb--Shanno algorithm with box
constraints (L-BFGS-B), are employed to obtain the MLEs.

\section{Simulation Study}

In this section, we evaluate  the performance of the MLEs and the asymptotic confidence intervals of the proposed model parameters. The primary objectives of the simulation are to assess the following:
\begin{itemize}
    \item The bias of $\hat{\beta}$: $ \text{Bias}(\hat{\beta})= \frac{1}{M}\sum_{i=1}^M (\hat{\beta}_i - \beta) $
    \item The mean squared error (MSE) of $\hat{\beta}$: $ \text{MSE}(\hat{\beta}) = \frac{1}{M} \sum_{i=1}^{M} \left( \hat{\beta}_i - \beta \right)^2 $
    \item Mean Relative Error (MRE) of $\hat{\beta}$: $\text{MRE}(\hat{\beta})=\frac{1}{M}\sum_{i=1}^M \left |\frac{\hat{\beta}_i - \beta}{\beta} \right | $
    \item The coverage probability (CP) for $\beta$:  $\text{CP}_{\beta}=\frac{1}{M}\sum_{i=1}^{M}I\left(\beta \in CI_i(\beta)\right) $
       \item The convergence rate (CR), defined as the proportion of simulation runs in which the optimization algorithm successfully converged,
\[
\text{CR} = \frac{\text{Number of successful estimation runs}}{\text{Total number of simulation runs}}
\]
 where $ \hat{\beta}_i $ is the MLE of $\beta$ (either $\omega$ or $\eta$) in the $ i$-th simulation run,  $i=1, \dots, M $, $CI_i(\beta)$ is the confidence interval of $\beta$ from the $ i$-th sample, and $I(\cdot)$  is the indicator function, i.e.,  $I(\alpha)=1$  if  $\alpha$ is true, 0 otherwise.
\end{itemize}
\noindent The simulation procedure is as follows (using the R software):
\begin{enumerate}
    \item Generate independent random samples of size $n = 30, 60, 100, 150$, and $200$ from the UShD using a rejection sampling algorithm based on the derived PDF.
    \item For each generated sample, estimate the parameters $(\omega, \eta)$ by maximizing the likelihood function using the L-BFGS-B optimization method with appropriate lower and upper bounds.
    \item For each estimated parameter, construct bootstrap-based confidence intervals with $B = 100$ resamples to approximate the sampling distribution and calculate the coverage probability.
    \item Repeat the above steps for $M = 1000$ simulation runs to obtain reliable estimates of Bias, MSE, MRE, CP, and CR.
\end{enumerate}
\noindent The simulation results are presented in Table~\ref{T:sim2}, showing the performance metrics for different sample sizes. The simulation study provides insight into the finite-sample behavior of the MLEs, illustrating how estimation accuracy and confidence interval coverage improve with increasing sample size.

\begin{table}[!ht]
\centering
\caption{Simulation results.} \label{T:sim2}
\adjustbox{max width=\textwidth, center}{\small
\begin{tabular}{l c c  c  c c  c c c c c c}\hline

    $\omega$& $\eta$& $n $&Bias$(\hat{\omega})$  & Bias$(\hat{\eta})$    &MSE$(\hat{\omega})$  &MSE$(\hat{\eta})$  & MRE$(\hat{\omega})$  &MRE$(\hat{\eta})$ &$\text{CP}_{\omega}$ &$\text{CP}_{\eta}$    &CR       \\ \hline

0.6  & 0.2  & 30 & 0.0305 & 0.0889 & 0.0139 & 0.0590 & 0.3019 & 0.7618 & 0.9456 & 0.9379 &0.972 \\
     &      & 60 & 0.0230 & 0.0840 & 0.0081 & 0.0535 & 0.2285 & 0.7066 & 0.9474 & 0.9395 &0.969 \\
     &      & 100 & 0.0185 & 0.0389 & 0.0045 & 0.0285 & 0.1707 & 0.5170 & 0.9419 & 0.9498 &0.973 \\
     &      & 150 & 0.0145 & 0.0303 & 0.0033 & 0.0262 & 0.1507 & 0.4784 & 0.9529 & 0.9476 &0.969 \\
     &      & 200 & 0.0124 & 0.0254 & 0.0024 & 0.0224 & 0.1288 & 0.4487 & 0.9498 & 0.9480 &0.978 \\ \hline

0.6  & 1.8 & 30 & 0.0401 & 0.1502 & 0.0144 & 0.7765 & 0.0993 & 0.4454 & 0.9465 & 0.9404 &0.976 \\
     &      & 60 & 0.0290 & 0.1157 & 0.0076 & 0.6970 & 0.0749 & 0.4080 & 0.9454 & 0.9428 &0.973 \\
     &      & 100 & 0.0170 & 0.0239 & 0.0041 & 0.1814 & 0.0548 & 0.1409 & 0.9415 & 0.9457 &0.983 \\
     &      & 150 & 0.0127 & 0.0190 & 0.0026 & 0.1068 & 0.0435 & 0.1013 & 0.9549 & 0.9485 &0.975 \\
     &      & 200 & 0.0114 & 0.0081 & 0.0019 & 0.0454 & 0.0377 & 0.0871 & 0.9545 & 0.9465 &0.980 \\ \hline

1    & 0.7  & 30 & 0.0582 & 0.2135 & 0.0424 & 0.5668 & 0.1882 & 0.7689 & 0.9418 & 0.9438 &0.975 \\
     &      & 60 & 0.0349 & 0.1939 & 0.0201 & 0.5018 & 0.1347 & 0.6930 & 0.9491 & 0.9436 &0.977 \\
     &      & 100 & 0.0260 & 0.0689 & 0.0125 & 0.2747 & 0.1060 & 0.5931 & 0.9531 & 0.9379 &0.962 \\
     &      & 150 & 0.0182 & 0.0505 & 0.0078 & 0.1354 & 0.0846 & 0.4375 & 0.9574 & 0.9418 &0.959 \\
     &      & 200 & 0.0144 & 0.0017 & 0.0067 & 0.0973 & 0.0768 & 0.3305 & 0.9474 & 0.9455 &0.967 \\    \hline

1.2  & 0.8  & 30 & 0.0686 & 0.0679 & 0.0572 & 0.3763 & 0.1948 & 0.6159 & 0.9309 & 0.9350 &0.974 \\
     &      & 60 & 0.0439 & 0.0397 & 0.0280 & 0.3584 & 0.1366 & 0.5947 & 0.9452 & 0.9426 &0.987 \\
     &      & 100 & 0.0264 & 0.0110 & 0.0155 & 0.2556 & 0.1000 & 0.4914 & 0.9483 & 0.9439 &0.973 \\
     &      & 150 & 0.0206 & 0.0103 & 0.0108 & 0.1594 & 0.0846 & 0.3589 & 0.9498 & 0.9427 &0.996 \\
     &      & 200 & 0.0165 & 0.0098 & 0.0095 & 0.0763 & 0.0804 & 0.2952 & 0.9469 & 0.9498 &0.985 \\ \hline

1.5  & 0.2  & 30 & 0.0305 & 0.0802 & 0.0805 & 0.0598 & 0.6346 & 0.6568 & 0.9471 & 0.9430 &0.988 \\
     &      & 60 & 0.0264 & 0.0486 & 0.0435 & 0.0519 & 0.4689 & 0.6072 & 0.9478 & 0.9407 &0.983 \\
     &      & 100 & 0.0200 & 0.0072 & 0.0265 & 0.0305 & 0.3663 & 0.4623 & 0.9504 & 0.9574 &0.992 \\
     &      & 150 & 0.0141 & 0.0049 & 0.0179 & 0.0275 & 0.2933 & 0.4337 & 0.9473 & 0.9564 &0.991 \\
     &      & 200 & 0.0091 & 0.0017 & 0.0116 & 0.0153 & 0.1599 & 0.2855 & 0.9476 & 0.9507 &0.995 \\ \hline

2    &  0.6 & 30 & 0.0851 & 0.0803 & 0.1461 & 0.2972 & 0.3322 & 0.5627 & 0.9396 & 0.9366 &0.996 \\
     &      & 60 & 0.0446 & 0.0788 & 0.0767 & 0.2329 & 0.2319 & 0.5338 & 0.9589 & 0.9548 &0.998 \\
     &      & 100 & 0.0390 & 0.0351 & 0.0474 & 0.1952 & 0.1872 & 0.4601 & 0.9476 & 0.9406 &0.995 \\
     &      & 150 & 0.0180 & 0.0527 & 0.0321 & 0.1205 & 0.1589 & 0.4183 & 0.9469 & 0.9368 &0.998 \\
     &      & 200 & 0.0072 & 0.0065 & 0.0142 & 0.0261 & 0.1244 & 0.3298 & 0.9579 & 0.9469 &0.998 \\ \hline

2    & 1.4  & 30 & 0.0666 & 0.0669 & 0.1237 & 0.6743 & 0.1700 & 0.5361 & 0.9493 & 0.9438 &0.991 \\
     &      & 60 & 0.0562 & 0.0610 & 0.0711 & 0.5084 & 0.1264 & 0.5098 & 0.9457 & 0.9486 &0.996 \\
     &      & 100 & 0.0472 & 0.0539 & 0.0480 & 0.3296 & 0.1043 & 0.4539 & 0.9600 & 0.9530 &1.000 \\
     &      & 150 & 0.0209 & 0.0170 & 0.0301 & 0.0692 & 0.0819 & 0.3410 & 0.9648 & 0.9557 &0.997 \\
     &      & 200 & 0.0080 & 0.0062 & 0.0247 & 0.0495 & 0.0720 & 0.2020 & 0.9509 & 0.9539 &0.999 \\ \hline

\end{tabular}}
\end{table}

\section{Applications}
This section aims to illustrate the effectiveness and flexibility of the USh model through its application to real-life data sets and its comparison with several alternative models. The goodness-of-fit of the proposed model is evaluated against seven competing models presented in Table ~\ref{T:dists}. Four real-life data sets are examined by fitting both the USh model and its competitors, and their performances are assessed using multiple information criteria, including the AIC, AICC, BIC, and HQIC. In addition, the Kolmogorov–Smirnov (KS) test is employed to measure the agreement between the empirical distribution and the fitted theoretical distributions, with the corresponding p-values used for model comparison. This test is widely adopted for evaluating how well a theoretical distribution represents observed data. In general, the model that yields the smallest values of the information criteria and the largest p-value is regarded as providing the best fit.

\begin{table}[!ht]
\centering
\caption{Some competitive models for the unit Shiha distribution } \label{T:dists}
\adjustbox{max width=\textwidth, center}{ \small
\renewcommand{\arraystretch}{1.5}
\begin{tabular}{l  l c } \hline

    Distribution    & Probability density function & Reference       \\ \hline

  Kumaraswamy (Kw) & $f_{Kw}(x) = \omega \eta x^{\eta - 1} (1 - x^{\eta})^{\omega - 1}\; , \quad \omega >0, \eta > 0 $ & \cite{kum} \\
  Two-parameter Unit Bilal (UB) & $f_{UB}(x)=\frac{6\omega}{\eta}\left(1 - x^{\omega / \eta}\right) x^{2\omega / \eta - 1}  ,\quad  \omega, \eta > 0$ & \cite{bilal} \\
  Unit Exponential (UE) & $f_{UE}(x) = \frac{2\omega\eta}{1-x^2} \left( \frac{1+x}{1-x} \right)^{\eta} \exp\left[ \omega \left( 1 - \left( \frac{1 + x}{1 - x} \right)^{\eta} \right) \right]  , \quad  \omega, \eta > 0$ & \cite{hasan}\\
   Exponentiated UEHL (EUEHL) & $f_{EUEHL}(x) = 2\omega\eta\alpha\frac{x^{\omega-1}}{\left(1+x^{\omega}\right)^{2}}\left( \frac{1-x^{\omega}}{1+x^{\omega}}\right)^{\eta-1} \left[ 1 - \left( \frac{1-x^{\omega}}{1+x^{\omega}} \right)^{\eta} \right]^{\alpha-1} \; , \quad \omega > 0,  \eta > 0, \; \alpha > 0$&\cite{genc} \\
  Unit Exponentiated Lomax (UEL) & $ f_{UEL}(x) = \frac{\eta \omega \alpha}{x} (1 - \eta \ln(x))^{-\omega-1} \left\{ 1 - (1 - \eta \ln(x))^{-\omega} \right\}^{\alpha-1}  , \quad \omega > 0, \; \eta > 0, \; \alpha > 0$ & \cite{fayomi}\\
  Beta & $ f_{Beta}(x) = \frac{\Gamma(\omega+\eta)}{\Gamma(\omega)\Gamma(\eta)}x^{\omega-1}(1-x)^{\eta-1}, \quad \omega, \eta > 0$ & \cite{beta} \\
  Topp Leone  (TL) & $f_{TL}(x) = 2\omega x^{\omega - 1} \left( 1 - x \right) \left( 2 - x \right)^{\omega - 1}  , \quad  \omega >0 $& \cite{topp}\\ \hline
\end{tabular}}
\end{table}
The data sets considered in this study have been obtained from different sources and have been widely used in the literature. To ensure consistency and suitability for modeling unit probability distributions, we selected data sets that are either naturally defined on the unit interval (0, 1) or can be appropriately transformed to this range. Some of the data sets correspond to capacity factors, which are defined as the ratio of actual electricity output to the maximum possible output of a power unit and therefore lie in the interval (0, 1). These data sets are commonly used as reliability measures for evaluating the efficiency of energy generation systems and have been analyzed using different estimation algorithms. In addition, the study includes lifetime and reliability data, such as medical survival times and Type I censored failure-time data, which are transformed to the unit interval for modeling purposes.
A detailed description of each data set is provided below.

\noindent \textbf{Data I}: The following data set consists of the lifetimes (in days) of 43 blood cancer patients from one of the Ministry of Health hospitals in Saudi Arabia \cite{abo}:

\noindent 115, 181, 255, 418, 441, 461, 516, 739, 743, 789, 807, 865, 924, 983, 1025, 1062, 1063, 1165, 1191, 1222, 1222, 1251, 1277, 1290, 1357, 1369, 1408, 1455, 1478, 1519, 1578, 1578, 1599, 1603, 1605, 1696, 1735, 1799, 1815 ,1852, 1899, 1925, 1965. Dividing each observation by 1970 yields data values between 0 and 1.

\noindent The following data sets II and III were reported by \cite{caram} and have been studied recently by many authors, for example \cite{eliwa, alz}. These data provide a comparison of two different algorithms, namely SC16 and P3, for estimating unit capacity factors.

\noindent \textbf{Data II}: The observations produced by the SC16 algorithm are:

\noindent 0.853, 0.759, 0.866, 0.809, 0.717, 0.544, 0.492, 0.403, 0.344, 0.213, 0.116, 0.116, 0.092,
0.070, 0.059, 0.048, 0.036, 0.029, 0.021, 0.014, 0.011, 0.008, 0.006.

\noindent \textbf{Data III}: The values obtained from the P3 algorithm are:

\noindent 0.853, 0.759, 0.874, 0.800, 0.716, 0.557, 0.503, 0.399, 0.334, 0.207, 0.118,
0.118, 0.097, 0.078, 0.067, 0.056, 0.044, 0.036, 0.026, 0.019, 0.014, 0.010.

\noindent \textbf{Data IV}: The type I censored data set, discussed by \cite{murt}, consists of 50 items that were tested, with the test terminated after 12 hours. The failure times observed before the termination of the test are as follows:

 \noindent 0.80, 1.26, 1.29, 1.85, 2.41, 2.47, 2.76, 3.35, 3.68, 4.46, 4.65, 4.83, 5.21, 5.26, 5.36, 5.39, 5.53, 5.64, 5.80, 6.08, 6.38, 7.02, 7.18, 7.60, 8.13,8.46, 8.69, 10.52, 11.25, 11.90. Dividing each observation by 12 yields data values between 0 and 1.

\noindent Table~\ref{tab:desc} presents some descriptive statistics of the data. Figures~\ref{fig:ttt1}--\ref{fig:ttt4} display the TTT, box, and violin plots of the four data sets.

\begin{table}[!ht]
\centering
\caption{Some descriptive statistics for the four data sets}
\label{tab:desc}
\begin{tabular}{c c c c c c c c c c}
\hline
Data& Min & Q1 & Median & Mean  & Q3 & Max  & Variance & SK & KU \\ \hline
I&0.058 & 0.424 & 0.635 & 0.605 & 0.806 & 0.998 & 0.066 & -0.442 & 5.279 \\
II&0.006 & 0.032 & 0.116 & 0.288 & 0.518 & 0.866 & 0.101 & 0.768 & 4.974 \\
III&0.01 & 0.047 & 0.118 & 0.304 & 0.544 & 0.874 & 0.101 & 0.711 & 4.884 \\
IV&0.067 & 0.286 & 0.448 & 0.459 & 0.595 & 0.992 & 0.058 & 0.388 & 5.645 \\ \hline
\end{tabular}
\end{table}
\begin{figure}[H]
  \centering
 \includegraphics[width=14.5cm, height=4cm]{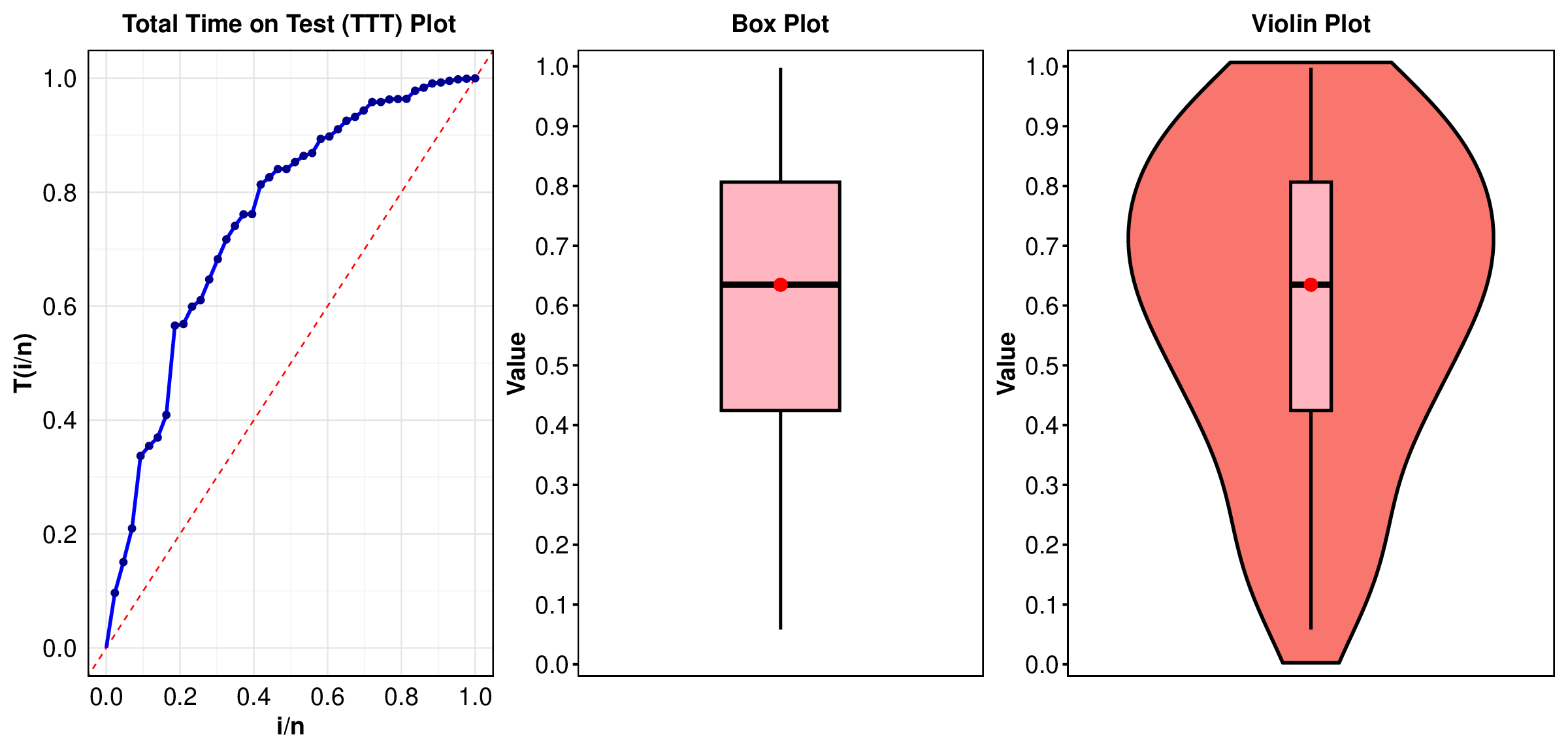}
\caption{TTT, box, and violin plots for data I.}
 \label{fig:ttt1}
\end{figure}
\begin{figure}[H]
  \centering
 \includegraphics[width=14.5cm, height=4cm]{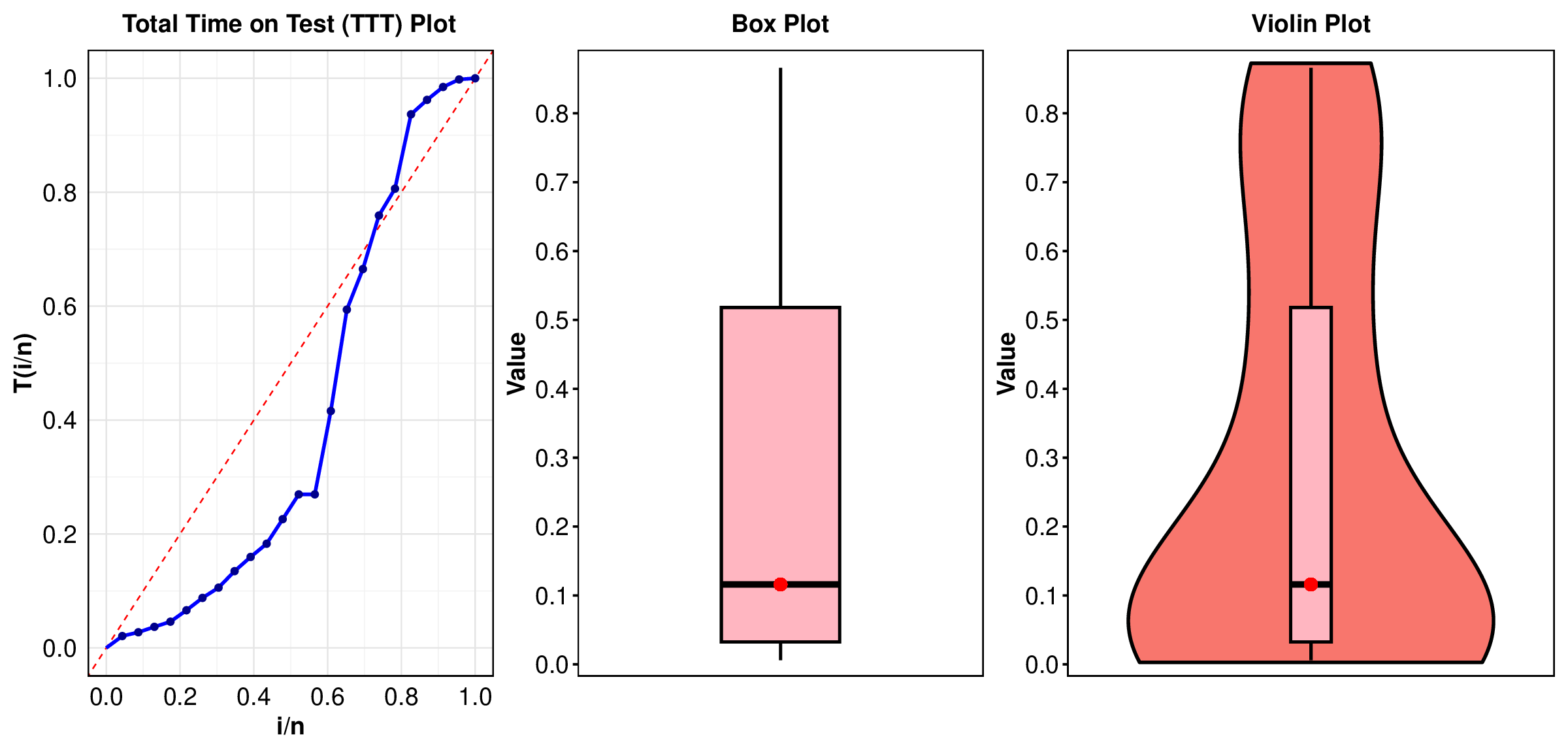}
\caption{TTT, box, and violin plots for data II.}
 \label{fig:ttt2}
\end{figure}
\begin{figure}[H]
  \centering
 \includegraphics[width=14.5cm, height=4cm]{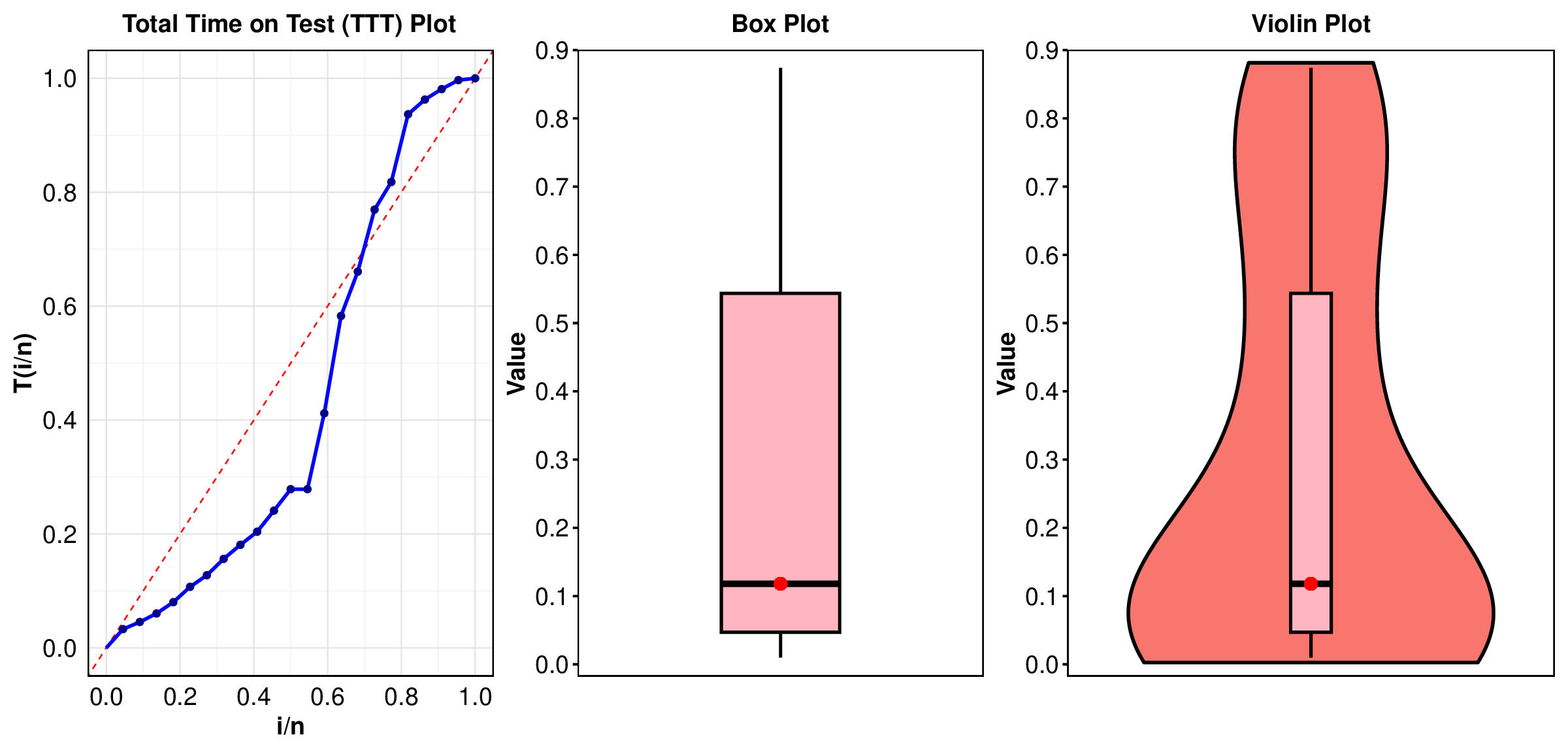}
\caption{TTT, box, and violin plots for data III.}
 \label{fig:ttt3}
\end{figure}
\begin{figure}[H]
  \centering
 \includegraphics[width=14.5cm, height=4cm]{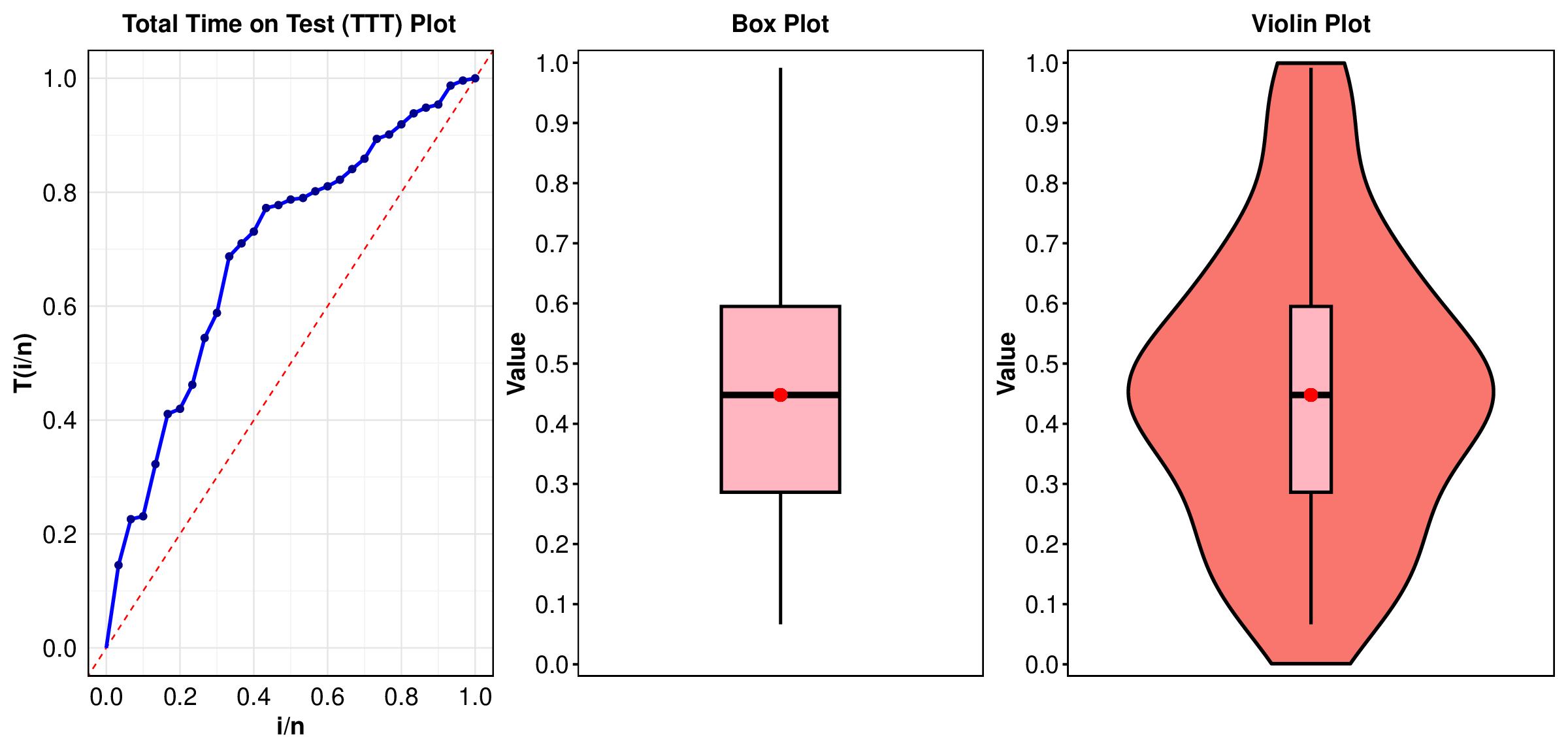}
\caption{TTT, box, and violin plots for data IV.}
 \label{fig:ttt4}
\end{figure}

\noindent Based on the descriptive statistics reported in Table~\ref{tab:desc} and the corresponding box and violin plots, it can be observed that data set I exhibits a left-skewed distribution, whereas data sets II, III, and IV are right-skewed. Furthermore, the TTT plots for data sets I and IV display a concave shape, indicating an increasing failure (hazard) rate. In contrast, the TTT plots for data sets II and III are initially convex and subsequently concave, suggesting a bathtub-shaped failure rate characterized by an initial decrease followed by an increase. Consequently, the USh distribution appears to be a suitable and flexible model for fitting these data sets.

\noindent Tables~\ref{T:app1}--\ref{T:app4} present the maximum likelihood estimates (MLEs) and goodness-of-fit information criteria for the four data sets across all competing models. The reported measures consistently indicate that the proposed USh distribution outperforms the alternative models, providing the best fit in most cases. These numerical results are further supported by graphical diagnostics. Specifically, by fitting each model to each data set, we assess their ability to capture the underlying distributional patterns and behaviors across different types of data. The histograms of the data sets overlaid with the fitted probability density functions, together with the empirical and estimated cumulative distribution functions, demonstrate the adequacy of the USh model. Moreover, the corresponding PP and QQ plots provide additional validation, where the close alignment of the USh-based points with the reference lines indicates a superior agreement with the observed data compared to the competing models. Overall, these results highlight the effectiveness and flexibility of the USh distribution, confirming its suitability for modeling diverse unit-valued data sets. See Figures~\ref{fig:fit1}--\ref{fig:qq4}.

\begin{table}[!ht]
\centering
\caption{MLEs and goodness-of-fit measures for Data I.}\label{T:app1}
\adjustbox{max width=\textwidth, center}{
\begin{tabular}{|l| c c c  c c c c c   c|}\hline
Distribution & $\hat{\omega}$ &$\hat{\eta}$ &$\hat{\alpha} $ & AIC  & AICC & BIC &HQIC   & KS & p-value \\ \hline
USh & 1.4957& 0.1221&-- & \textbf{-3.1059} & \textbf{-2.8059} & \textbf{0.4165} & \textbf{-1.8069}  & \textbf{0.0726} & \textbf{0.9772} \\
Kw & 1.0177& 1.559&-- & -3.0158 & -2.7158 & 0.5066 & -1.7169  & 0.0813 & 0.9390 \\
UB & 1.2819& 0.9874&-- & 11.4709 & 11.7709 & 14.9933 & 12.7699  & 0.1457 & 0.3205 \\
UE & 1.7603& 0.2293&-- & 2.5036 & 2.8036 & 6.0260 & 3.8026  & 0.1317 & 0.4447 \\
EUEHL & 0.075& 1.0012& 40.3039 & -1.0279 & -0.4125 & 4.2557 & 0.9205 & 0.0799 & 0.9465 \\
UEL & 54.3925& 0.0294& 1.0319 & -0.9877 & -0.3724 & 4.2959 & 0.9607  & 0.0819 & 0.9350 \\
Beta & 1.5613& 1.0173&-- & -3.0159 & -2.7159 & 0.5065 & -1.7170  & 0.0812 & 0.9393 \\
TL & 3.0122&--&-- & 15.9571 & 16.0546 & 17.7183 & 16.6065  & 0.1877 & 0.0965 \\ \hline
\end{tabular}}
\end{table}
\begin{table}[!ht]
\centering
\caption{MLEs and goodness-of-fit measures for Data II.}\label{T:app2}
\adjustbox{max width=\textwidth, center}{
\begin{tabular}{|l| c c c  c c c c c   c|}\hline
Distribution & $\hat{\omega}$ &$\hat{\eta}$ &$\hat{\alpha} $ & AIC  & AICC & BIC &HQIC  & KS & p-value \\ \hline
USh & 0.374& 193.1469&-- & \textbf{-15.7077} & \textbf{-15.1077} & \textbf{-13.4367} & \textbf{-15.1366}  & \textbf{0.1512} & \textbf{0.6694} \\
Kw & 1.1862& 0.5044&-- & -15.3416 & -14.7416 & -13.0706 & -14.7704  & 0.1790 & 0.4529 \\
UB & 2.0102& 5.3038&-- & -11.2158 & -10.6158 & -8.9448 & -10.6446  & 0.1786 & 0.4553 \\
UE & 107.6512& 0.0126&-- & -9.0515 & -8.4515 & -6.7805 & -8.4804  & 0.2941 & 0.0374 \\
EUEHL & 0.3132& 1.0549& 2.4651 & -14.6813 & -13.4181 & -11.2748 & -13.8246  & 0.1584 & 0.6112 \\
UEL & 52.1005& 0.0099& 1.2034 & -13.1600 & -11.8968 & -9.7535 & -12.3033  & 0.1808 & 0.4400 \\
Beta & 0.4869& 1.1679&-- & -15.2149 & -14.6149 & -12.9439 & -14.6438  & 0.1836 & 0.4202 \\
TL & 0.5943& -- &-- & -14.2302 & -14.0398 & -13.0948 & -13.9447 & 0.1690 & 0.5272 \\ \hline
\end{tabular}}
\end{table}
\begin{table}[!ht]
\centering
\caption{MLEs and goodness-of-fit measures for Data III.}\label{T:app3}
\adjustbox{max width=\textwidth, center}{
\begin{tabular}{|l| c c c  c c c c c  c|}\hline
Distribution & $\hat{\omega}$ &$\hat{\eta}$ &$\hat{\alpha} $ & AIC  & AICC & BIC &HQIC  & KS & p-value \\ \hline
USh & 0.4159& 1000&-- & \textbf{-10.1225} & \textbf{-9.4910 }& \textbf{-7.9405} & \textbf{-9.6085} & \textbf{0.1760 }& \textbf{0.5029} \\
Kw & 1.2305& 0.5718&-- & -9.6872 & -9.0557 & -7.5051 & -9.1732  & 0.1963 & 0.3650 \\
UB & 1.7316& 4.1212&-- & -6.3672 & -5.7357 & -4.1852 & -5.8532  & 0.1827 & 0.4546 \\
UE & 93.4921& 0.0137&-- & -6.7942 & -6.1626 & -4.6121 & -6.2802  & 0.2824 & 0.0598 \\
EUEHL & 0.34& 1.0911& 2.6393 & -9.0181 & -7.6847 & -5.7449 & -8.2470 & 0.1788 & 0.4828 \\
UEL & 54.3238& 0.0107& 1.2387 & -7.5177 & -6.1844 & -4.2446 & -6.7467  & 0.1968 & 0.3616 \\
Beta & 0.554& 1.2198&-- & -9.5638 & -8.9323 & -7.3818 & -9.0498  & 0.2002 & 0.3413 \\
TL & 0.6778&-- &-- & -8.9965 & -8.7965 & -7.9055 & -8.7395  & 0.1848 & 0.4401 \\ \hline

\end{tabular}}
\end{table}

\begin{table}[!ht]
\centering
\caption{MLEs and goodness-of-fit measures for Data IV.}\label{T:app4}
\adjustbox{max width=\textwidth, center}{
\begin{tabular}{|l| c c c  c c c  c c  c|}\hline
Distribution & $\hat{\omega}$ &$\hat{\eta}$ &$\hat{\alpha} $ & AIC  & AICC & BIC &HQIC   & KS & p-value \\ \hline
USh & 0.8854& 545.452&-- & \textbf{-0.2354} & \textbf{0.2090} & 2.5670 & \textbf{0.6611}  & 0.1296 & 0.6478 \\
Kw & 1.4896 &1.3281&-- & 1.5721 & 2.0166 & 4.3745 & 2.4686 & 0.1304 & 0.6401 \\
UB & 0.9848& 1.1287&-- & 2.9225 & 3.3670 & 5.7249 & 3.8190  & \textbf{0.0881} & \textbf{0.9581} \\
UE & 7.8993& 0.0945&-- & 4.6855 & 5.1300 & 7.4879 & 5.5820  & 0.1540 & 0.4318 \\
EUEHL & 1.2418& 1.2371& 1.4679 & 1.4424 & 2.3655 & 5.6460 & 2.7872  & 0.1020 & 0.8830 \\
UEL & 91.3279& 0.0148& 1.501 & 3.7091 & 4.6322 & 7.9127 & 5.0539  & 0.1316 & 0.6287 \\
Beta & 1.3856& 1.5036& & 1.4478 & 1.8922 & 4.2502 & 2.3443  & 0.1288 & 0.6548 \\
TL & 1.8705 & --  &  --  & 1.0529 & 1.1957 & \textbf{2.4541} & 1.5011  & 0.0909 & 0.9461 \\ \hline

\end{tabular}}
\end{table}

\begin{figure}[H]
  \centering
 \includegraphics[width=14cm, height=4.5cm]{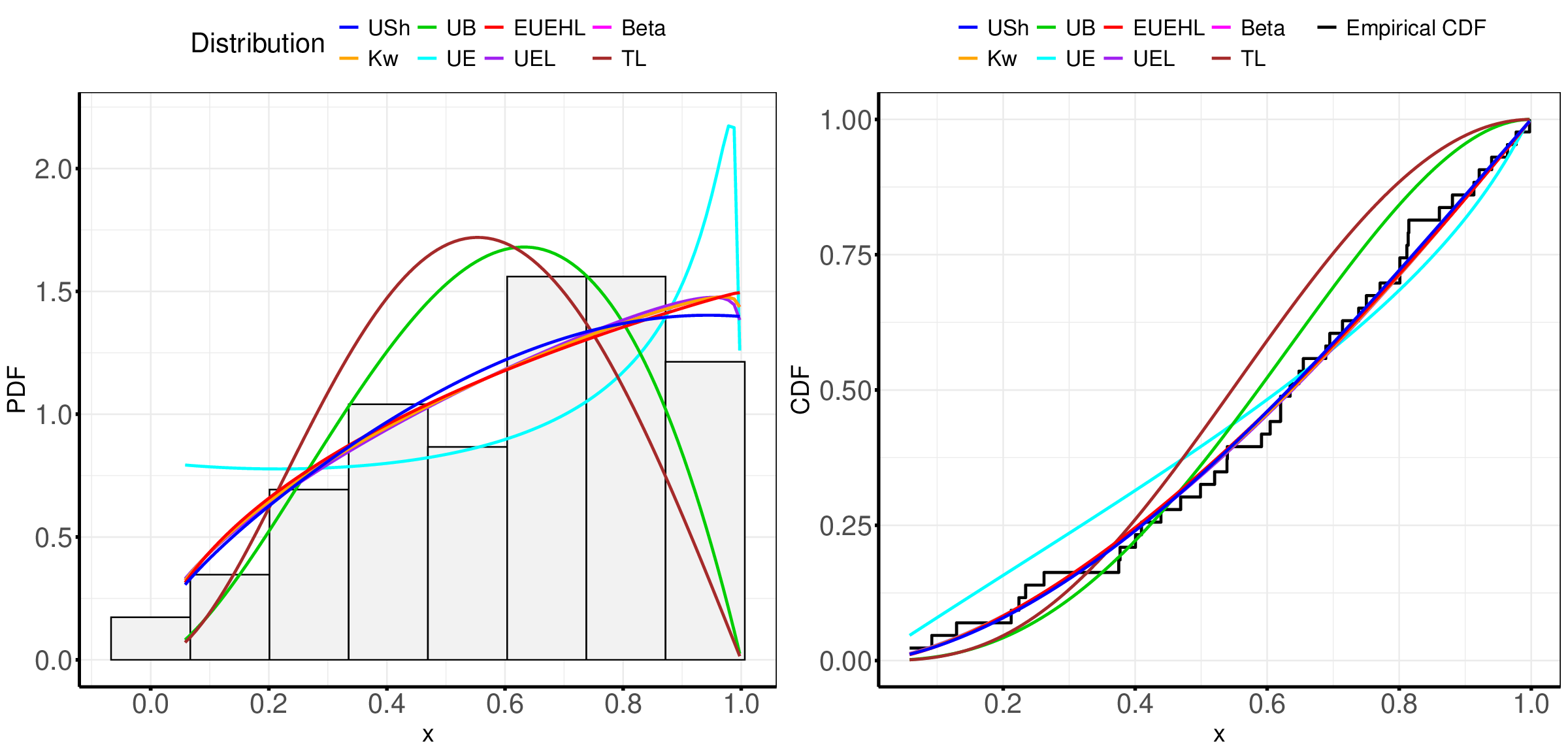}
  \caption{Fitted PDF and CDF of the competing models for Data I.}
 \label{fig:fit1}
\end{figure}
\begin{figure}[H]
  \centering
 \includegraphics[width=14cm, height=7cm]{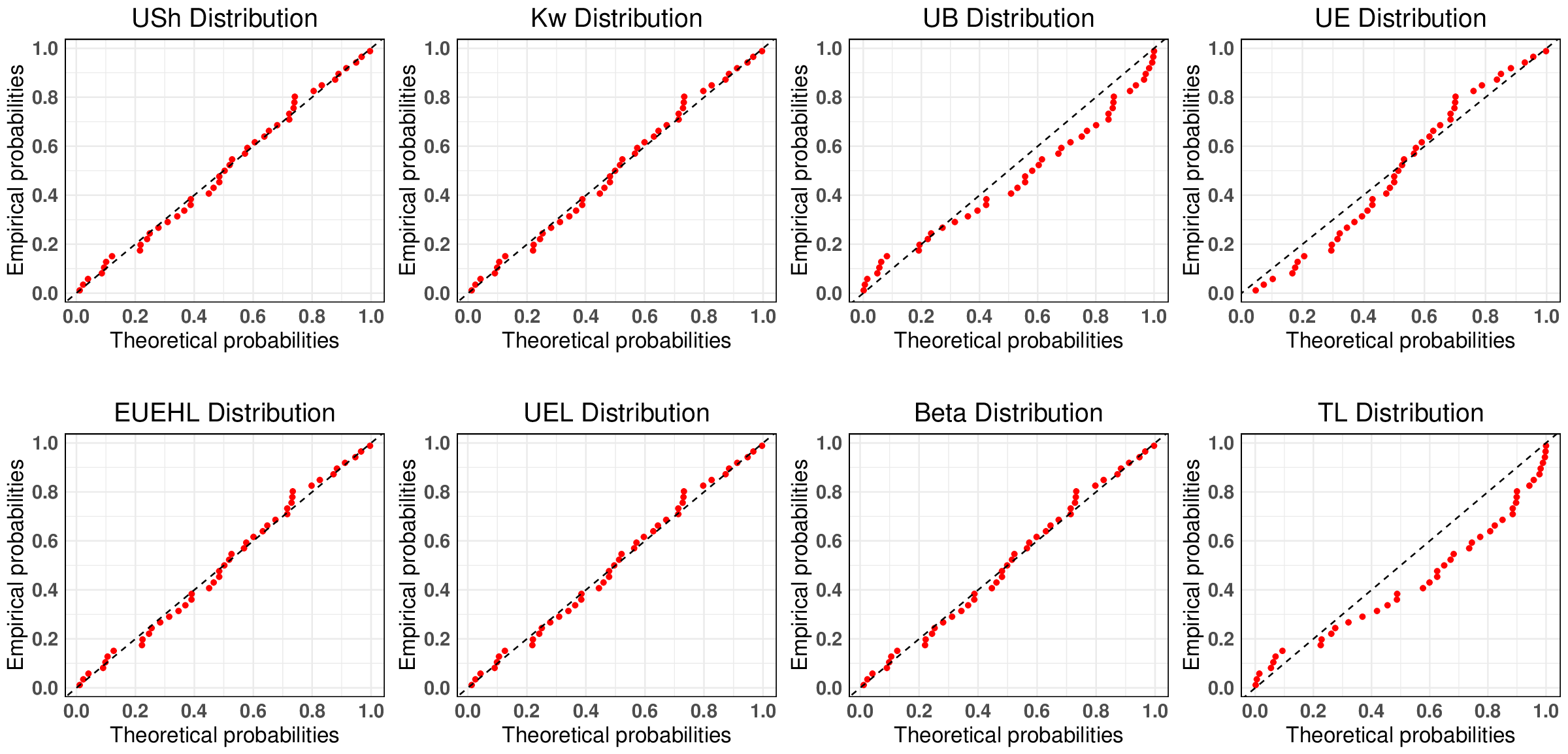}
\caption{Probability–probability (PP) plots of the competing models for Data I.}
 \label{fig:pp1}
\end{figure}
\begin{figure}[H]
  \centering
 \includegraphics[width=14cm, height=7cm]{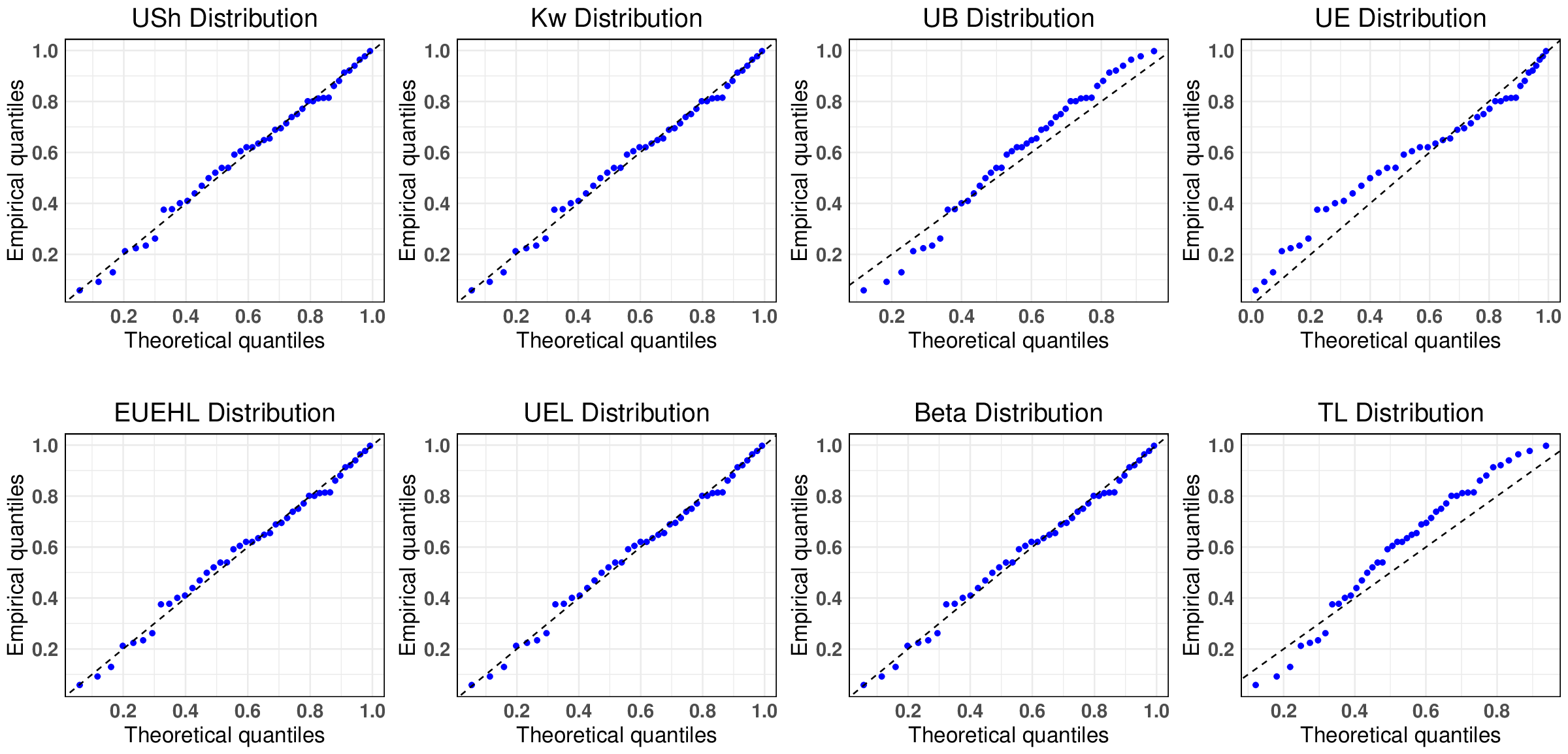}
  \caption{Quantile–quantile (QQ) plots of the competing models for Data I.}
 \label{fig:qq1}
\end{figure}

\begin{figure}[H]
  \centering
 \includegraphics[width=14cm, height=4.5cm]{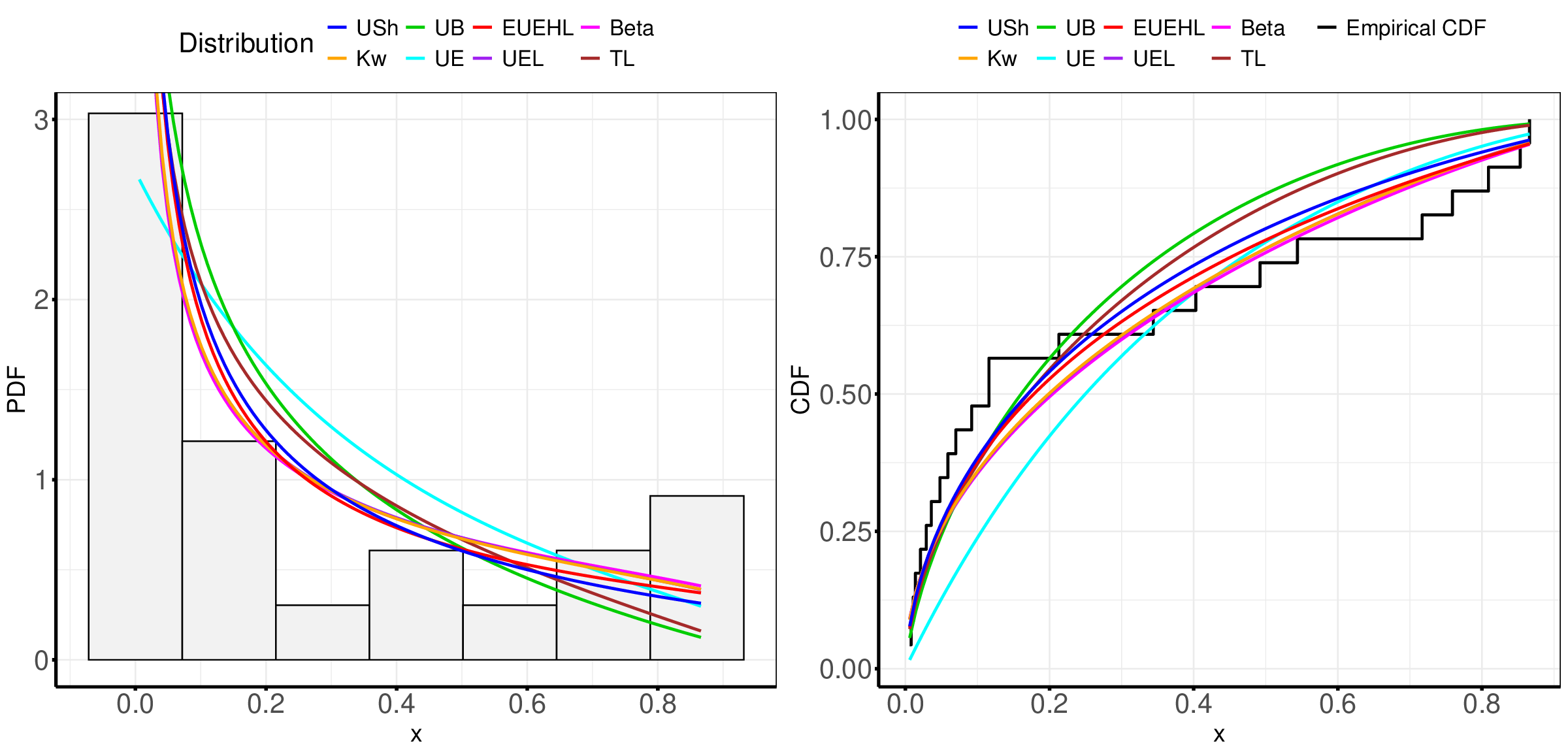}
  \caption{Fitted PDF and CDF of the competing models for Data II.}
 \label{fig:fit2}
\end{figure}
\begin{figure}[H]
  \centering
 \includegraphics[width=14cm, height=7cm]{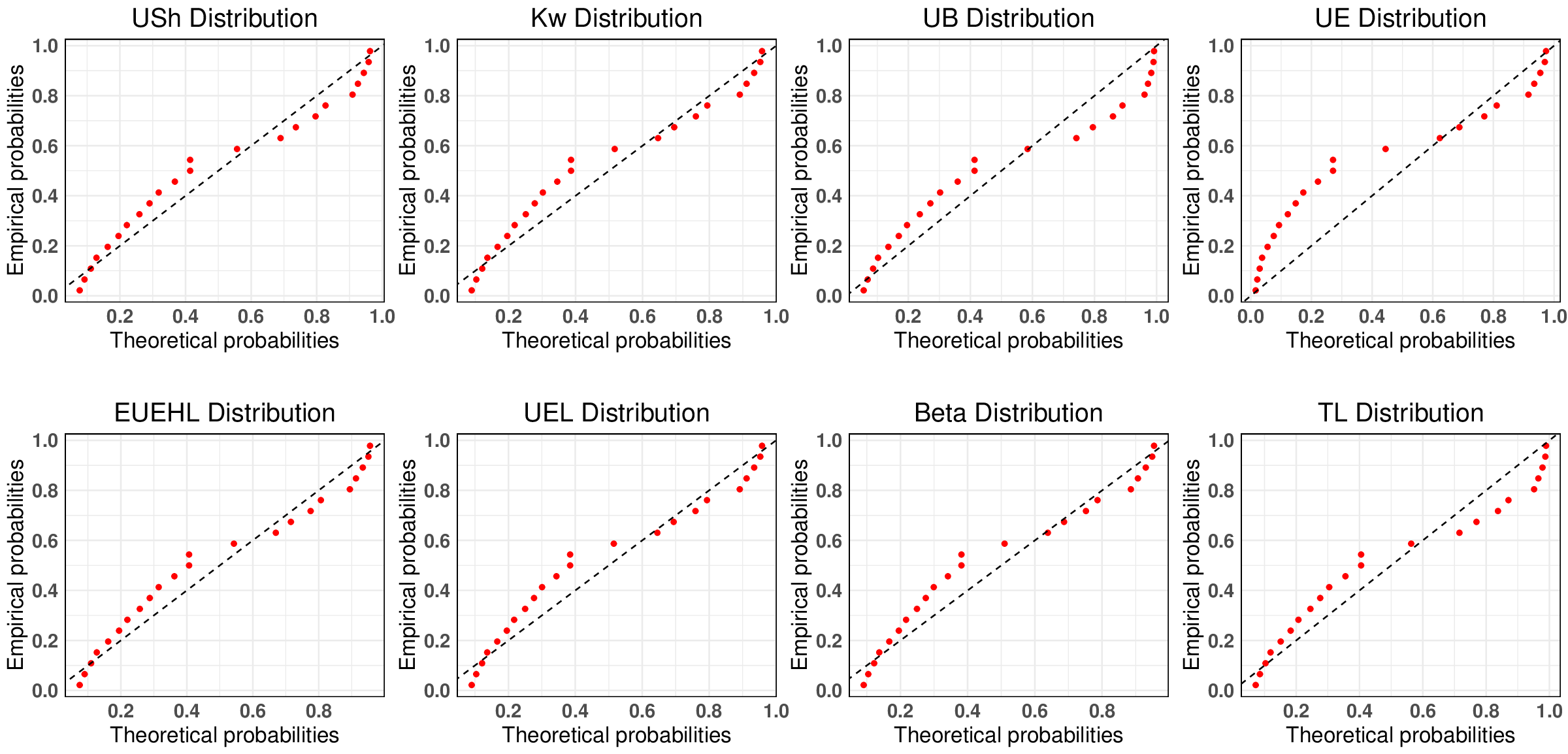}
\caption{Probability–probability (PP) plots of the competing models for Data II.}
 \label{fig:pp2}
\end{figure}
\begin{figure}[H]
  \centering
 \includegraphics[width=14cm, height=7cm]{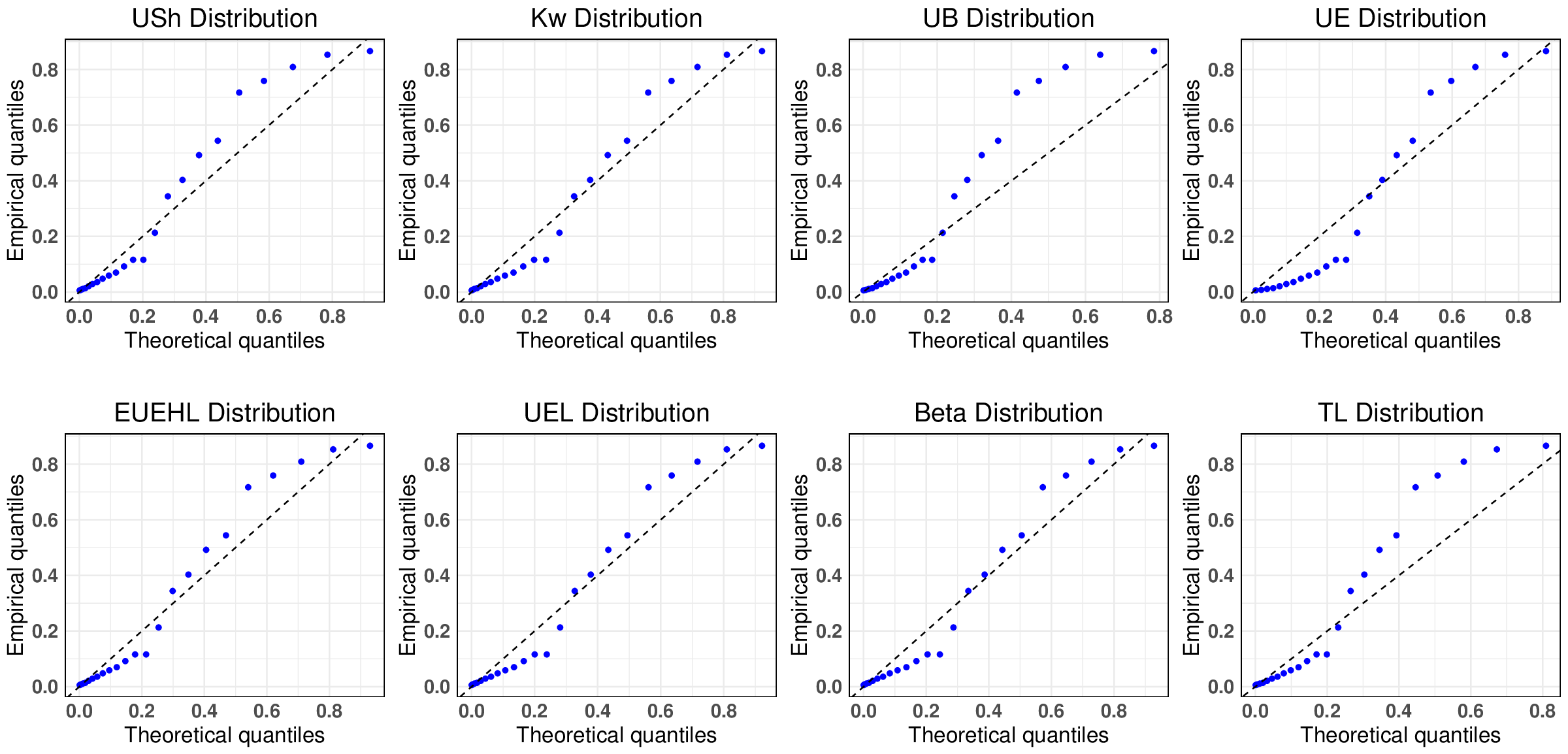}
  \caption{Quantile–quantile (QQ) plots of the competing models for Data II.}
 \label{fig:qq2}
\end{figure}

\begin{figure}[H]
  \centering
 \includegraphics[width=14cm, height=4.5cm]{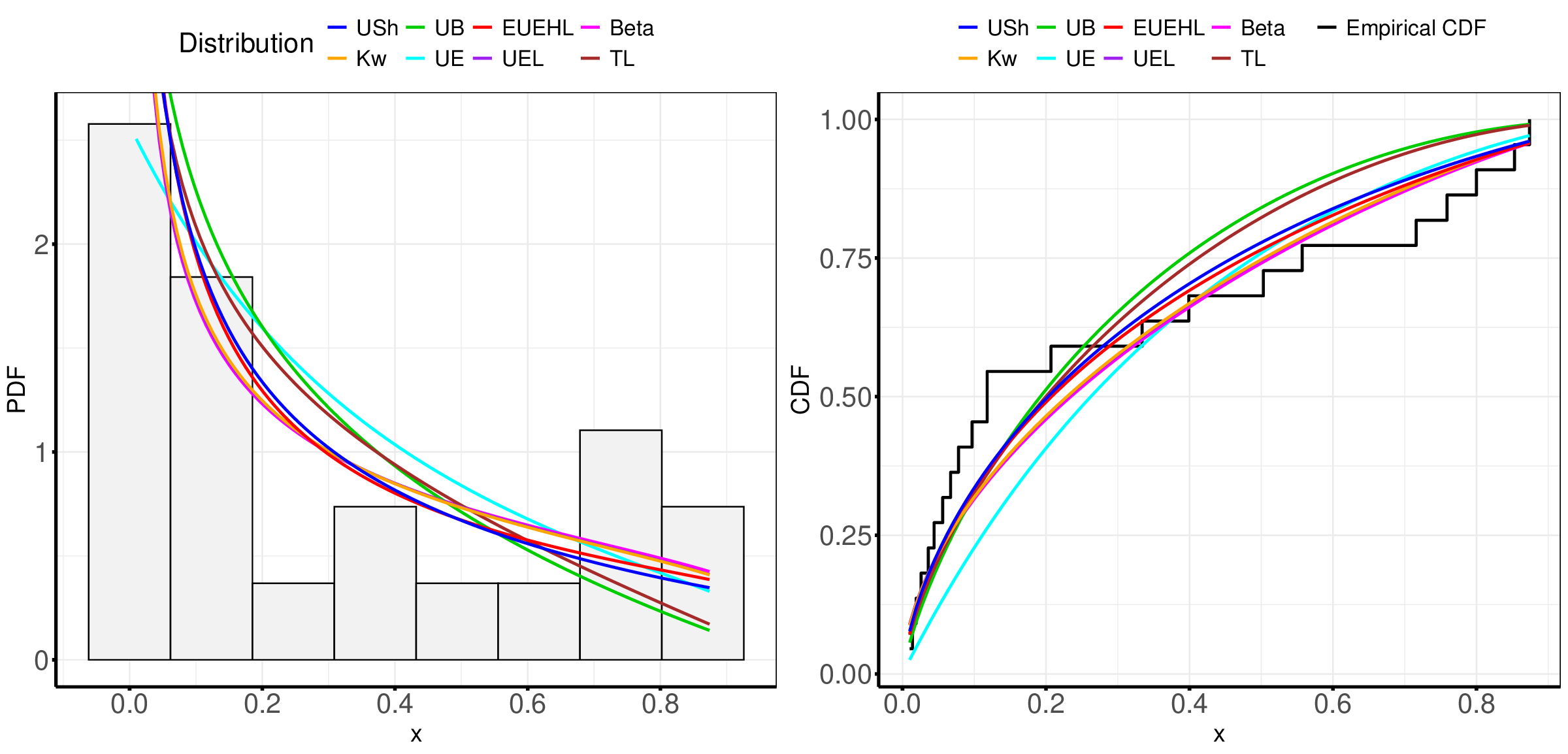}
  \caption{Fitted PDF and CDF of the competing models for Data III.}
 \label{fig:fit3}
\end{figure}
\begin{figure}[H]
  \centering
 \includegraphics[width=14cm, height=7cm]{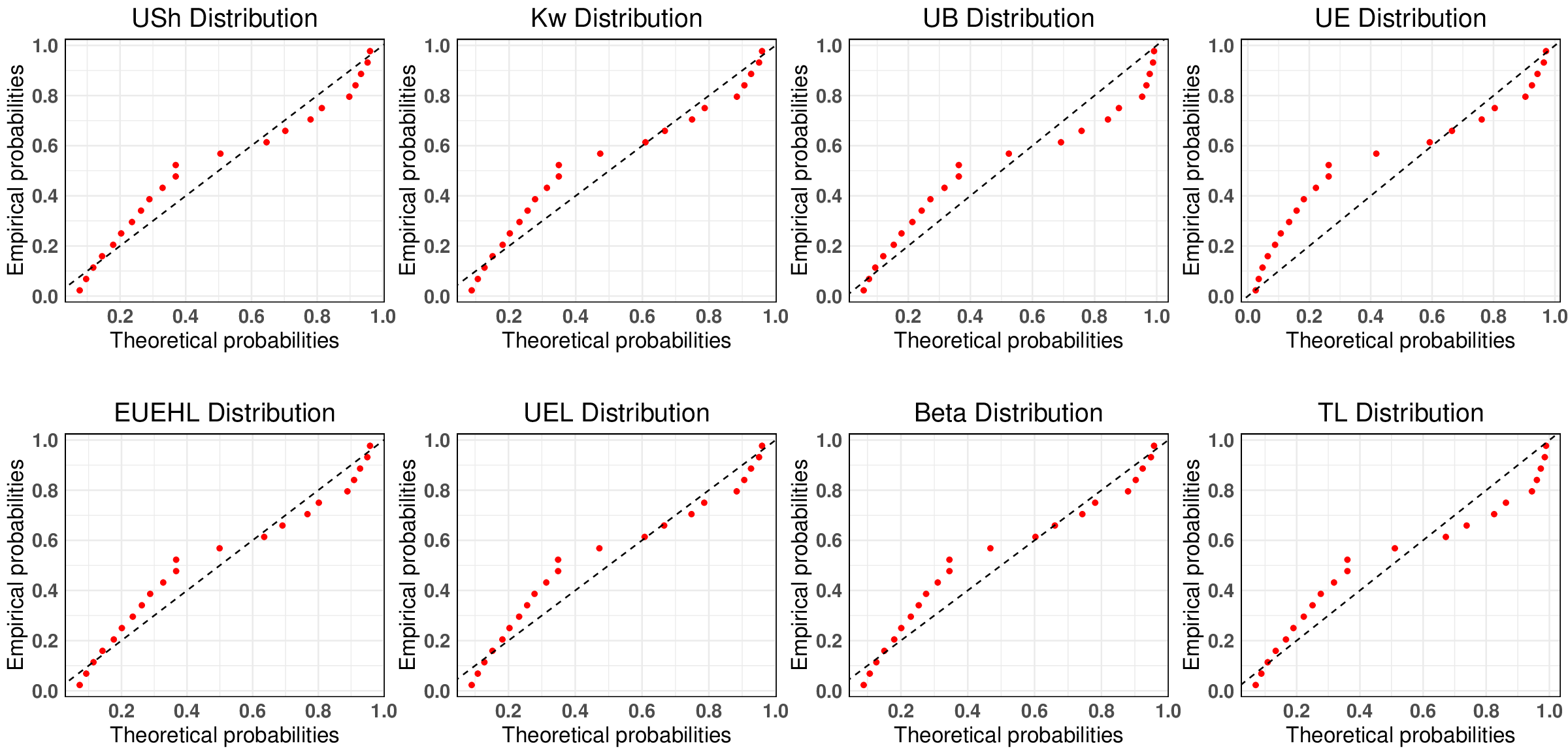}
\caption{ Probability–probability (PP) plots of the competing models for Data III.}
 \label{fig:pp3}
\end{figure}
\begin{figure}[H]
  \centering
 \includegraphics[width=14cm, height=7cm]{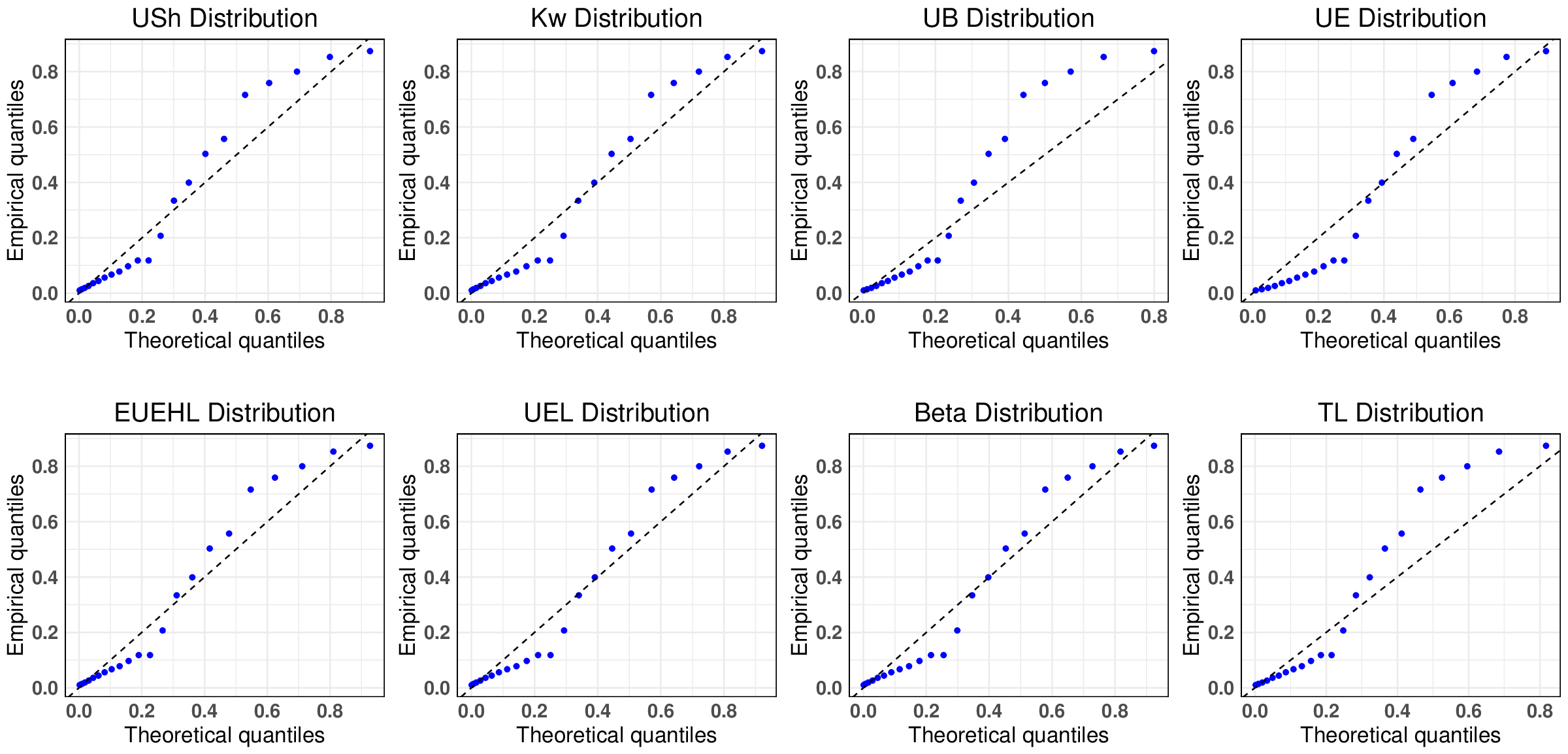}
  \caption{Quantile–quantile (QQ) plots of the competing models for Data III.}
 \label{fig:qq3}
\end{figure}

\begin{figure}[H]
  \centering
 \includegraphics[width=14cm, height=4.5cm]{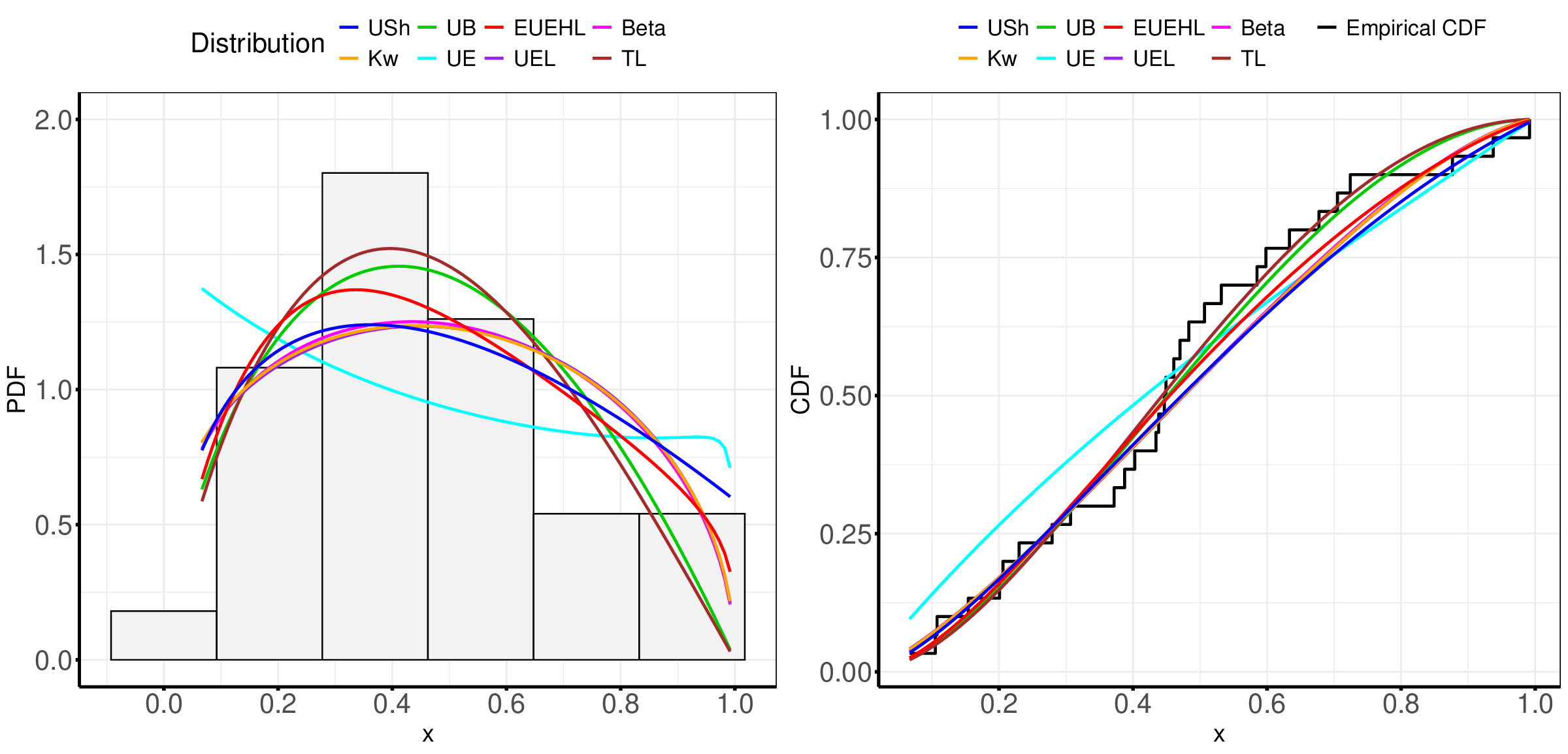}
  \caption{ Fitted PDF and CDF of the competing models for Data IV.}
 \label{fig:fit4}
\end{figure}
\begin{figure}[H]
  \centering
 \includegraphics[width=14cm, height=7cm]{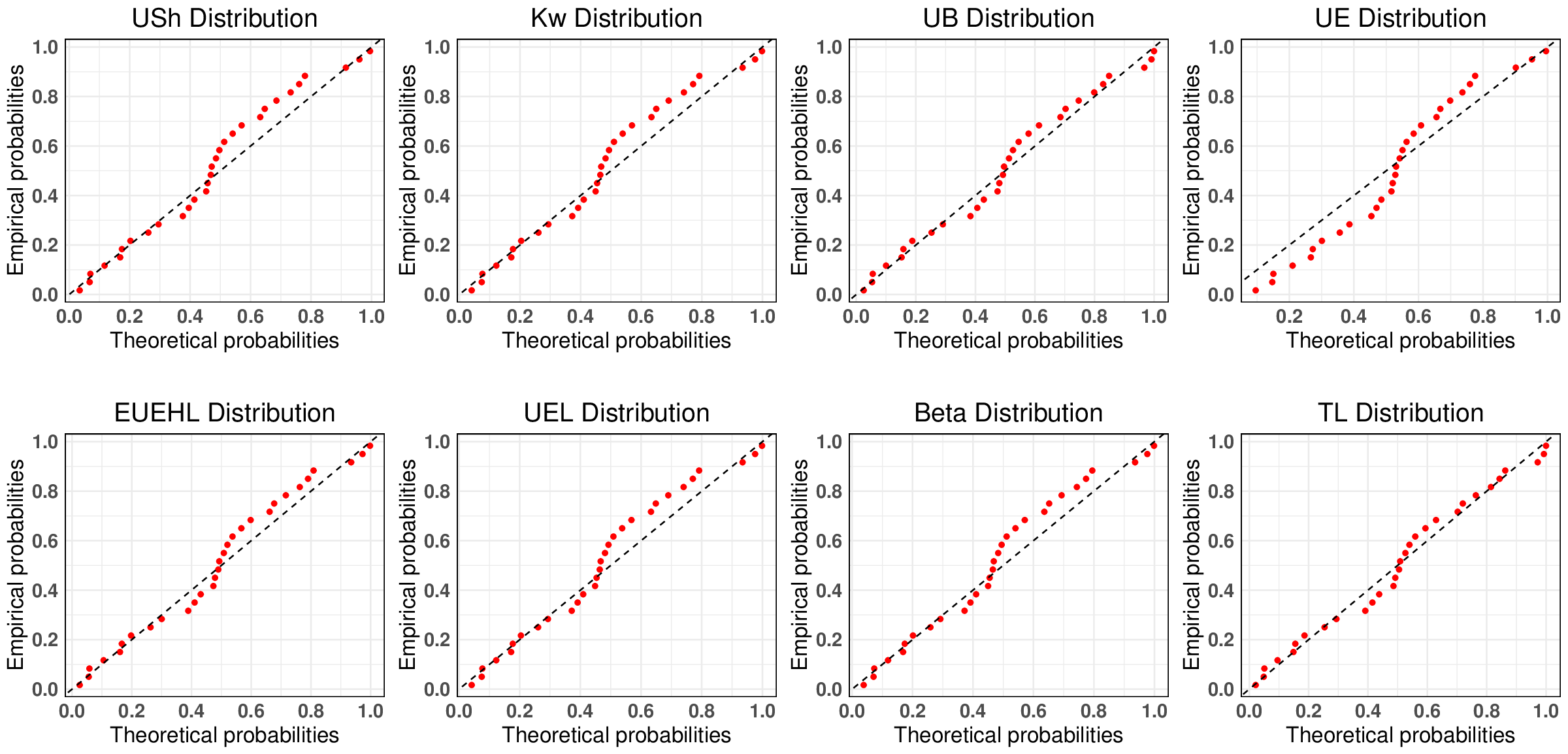}
\caption{ Probability–probability (PP) plots of the competing models for Data IV.}
 \label{fig:pp4}
\end{figure}
\begin{figure}[H]
  \centering
 \includegraphics[width=14cm, height=7cm]{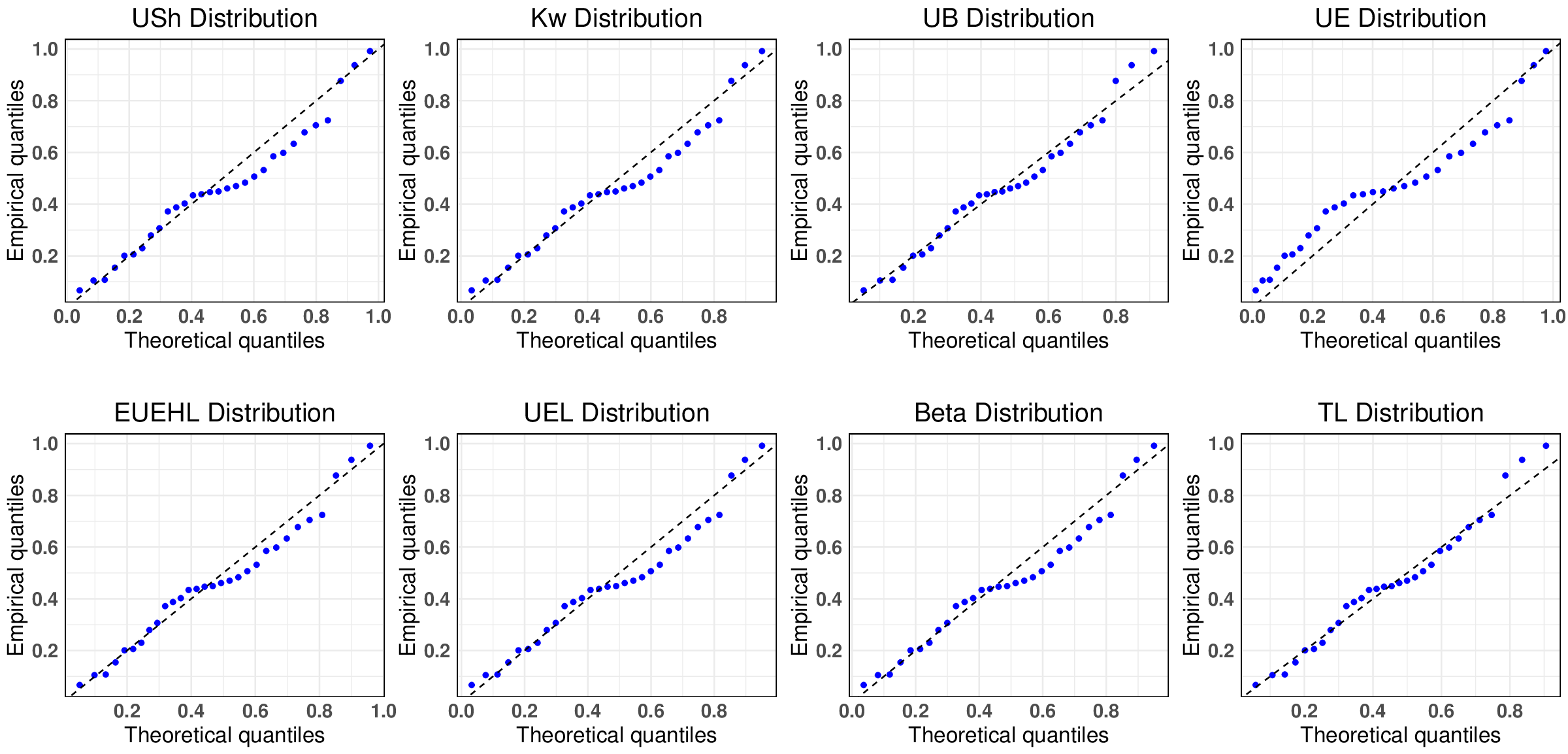}
  \caption{Quantile–quantile (QQ) plots of the competing models for Data set IV.
}
 \label{fig:qq4}
\end{figure}

\section{Conclusions}
This paper proposes a new unit distribution, termed the unit Shiha (USh) distribution, which is obtained through an inverse exponential transformation of the Sh distribution. Several statistical properties of the proposed model were derived, including moments, the quantile function, entropy, and stress–strength reliability. Owing to the flexibility of its probability density and hazard rate functions, the USh distribution is capable of modeling both left- and right-skewed data and representing different failure-rate behaviors, such as increasing and bathtub-shaped patterns.
MLE was employed for parameter inference, and a  simulation study was conducted to assess the finite-sample performance of the estimators. The simulation results demonstrate satisfactory estimation accuracy, good convergence behavior, and adequate coverage probabilities for the associated confidence intervals, with clear improvements as the sample size increases.
The practical applicability of the proposed model was illustrated using four real-life data sets and compared with several well-known competing unit distributions. Both the goodness-of-fit criteria and graphical diagnostics indicate that, in most cases, the USh distribution provides a superior fit to the considered data sets. Overall, the proposed USh distribution offers a flexible and effective modeling framework for unit-valued data and represents a valuable addition to the family of unit probability distributions.



\section*{Conflict of interest}
The authors declare that there is no conflict of interests regarding the publication of this paper.

\end{document}